\documentclass[runningheads]{llncs}
\usepackage{amsmath}%
\usepackage{amsfonts}%
\usepackage{amssymb}%
\usepackage{stmaryrd}%
\usepackage{float}%
\usepackage{array}%
\usepackage{wrapfig}
\usepackage{url}
\usepackage{mathtools}
\usepackage{cancel}

\newtheorem{fact}{Fact}
\DeclareMathAlphabet{\mathrm}    {OT1}{cmr}{m}{n}
\DeclareMathAlphabet{\mathrmbf}  {OT1}{cmr}{bx}{n}
\DeclareMathAlphabet{\mathrmit}  {OT1}{cmr}{m}{it}
\DeclareMathAlphabet{\mathrmbfit}{OT1}{cmr}{bx}{it}
\DeclareMathAlphabet{\mathsf}    {OT1}{cmss}{m}{n}
\DeclareMathAlphabet{\mathsfbf}  {OT1}{cmss}{bx}{n}
\DeclareMathAlphabet{\mathsfit}  {OT1}{cmss}{m}{sl}
\DeclareMathAlphabet{\mathttbf}  {OT1}{cmtt}{bx}{n}
\DeclareMathAlphabet{\mathttit}  {OT1}{cmtt}{m}{it}

\newcommand{\keywords}[1]{\par\addvspace\baselineskip\noindent\enspace\ignorespaces{\bfseries Keywords:\,}#1}
\newcommand{\comment}[1]{}
\begin{document}

\pagestyle{headings}
\title{The First-order Logical Environment}
\titlerunning{FOL}  
\author{Robert E. Kent}
\institute{Ontologos}
\maketitle

\begin{abstract}
This paper describes the first-order logical environment {\ttfamily FOLE}.
Institutions in general (Goguen and Burstall~\cite{goguen:burstall:92}),
and logical environments in particular,
give equivalent heterogeneous and homogeneous representations for logical systems.
As such,
they offer a rigorous and principled approach to distributed interoperable information systems
via system consequence (Kent~\cite{kent:iccs2009}).
Since {\ttfamily FOLE} is a particular logical environment, 
this provides a rigorous and principled approach to distributed interoperable first-order information systems.
The {\ttfamily FOLE} represents the formalism and semantics of first-order logic in a classification form. 
By using an interpretation form,
a companion approach (Kent~\cite{kent:db:sem})
defines the formalism and semantics of first-order logical/relational database systems.
In a strict sense,
the two forms have transformational passages (generalized inverses) between one another.
The classification form of first-order logic in the {\ttfamily FOLE} 
corresponds to ideas discussed in the Information Flow Framework (IFF~\cite{iff}).
The {\ttfamily FOLE} representation follows a conceptual structures approach,
that is completely compatible 
with formal concept analysis (Ganter and Wille~\cite{ganter:wille:99}) 
and information flow (Barwise and Seligman~\cite{barwise:seligman:97}).

\end{abstract}
%





\section{Introduction}\label{sec:intro}

The paper ``System Consequence'' (Kent~\cite{kent:iccs2009})
gave a general and abstract solution to the interoperation of information systems
via the channel theory of information flow (Barwise and Seligman~\cite{barwise:seligman:97}). 
These can be expressed either formally, semantically or in a combined form.
This general solution closely follows the theories of 
institutions (Goguen and Burstall~\cite{goguen:burstall:92}), 
\footnote{The technical aspect of this paper is described 
in the spirit of  
Goguen's categorical manifesto~\cite{goguen:91}
by using the terminology of
mathematical context, passage and bridge
in place of 
category, functor and natural transformation.}
information flow and 
formal concept analysis (Ganter and Wille~\cite{ganter:wille:99}). 
By following the approach of the ``System Consequence'' paper,
this paper offers a solution to the interoperation of 
distributed systems expressed in terms of the formalism and semantics of first-order logic. 
It does this be defining {\ttfamily FOLE},
the first-order logical environment.
%
\footnote{A logical environment is a special and more structurally pleasing case of an institution, 
where the semantics is completely compatible with satisfaction.}
%
Since this paper develops a classification form of first order logic as a logical environment,  
the interaction of information systems expressed in first order logic have a firm foundation.
Section~\ref{sec:arch} surveys the architecture of the first-order logical environment {\ttfamily FOLE}.
Section~\ref{sec:cmps} discusses the linguistic/formal and semantic components of {\ttfamily FOLE};
detailed discussions of the functional base and relational superstructure are given in
Appendix~\ref{sec:func-base} and Appendix~\ref{sec:rel-sup}, respectively. 
Section~\ref{sec:log-env} explains how {\ttfamily FOLE} is a logical environment;
a proof of this fact is given in Appendix~\ref{append:log-env}.
Section~\ref{sec:inf-sys} discusses {\ttfamily FOLE} information systems.
Finally,
section~\ref{sec:sum-fut-wrk} summarizes and states future plans for work on these topics. 

\section{Architecture}\label{sec:arch}

%
\begin{figure}
\begin{center}
{\begin{tabular}{c}
\setlength{\unitlength}{0.6pt}
\begin{picture}(360,240)(-115,-20)
\put(-90,30){\begin{picture}(0,0)(0,0)
\put(120,180){\makebox(0,0){\footnotesize{$\mathrmbf{Log}$}}}
\put(85,155){\makebox(0,0)[r]{\scriptsize{$\mathrmbfit{struc}$}}}
\put(110,170){\vector(-1,-1){40}}
\put(60,120){\makebox(0,0){\footnotesize{$\mathrmbf{Struc}$}}}
\put(0,60){\makebox(0,0){\footnotesize{$\mathrmbf{Rel}$}}}
\put(120,60){\makebox(0,0){\footnotesize{$\mathrmbf{Alg}$}}}
\put(60,0){\makebox(0,0){\footnotesize{$\mathrmbf{Cls}$}}}
\put(-20,50){\oval(20,20)[l]}
\put(-10,50){\oval(20,20)[br]}
\qbezier(-20,40)(-15,40)(-10,40)
\put(0,55){\vector(0,1){0}}
\put(-15,30){\makebox(0,0){\scriptsize{$\mathrmbfit{fmla}$}}}
\put(20,90){\makebox(0,0)[r]{\scriptsize{$\mathrmbfit{rel}$}}}
\put(100,90){\makebox(0,0)[l]{\scriptsize{$\mathrmbfit{alg}$}}}
\put(70,110){\vector(1,-1){40}}
\put(50,110){\vector(-1,-1){40}}
\put(10,50){\vector(1,-1){40}}
\put(110,50){\vector(-1,-1){40}}
\qbezier(50,100)(55,95)(60,90)
\qbezier(70,100)(65,95)(60,90)
\end{picture}}
\put(90,-15){\begin{picture}(0,0)(0,0)
\put(120,180){\makebox(0,0){\footnotesize{$\mathrmbf{Spec}$}}}
\put(85,155){\makebox(0,0)[r]{\scriptsize{$\mathrmbfit{lang}$}}}
\put(110,170){\vector(-1,-1){40}}
\put(60,120){\makebox(0,0){\footnotesize{$\mathrmbf{Lang}$}}}
\put(0,60){\makebox(0,0){\footnotesize{$\mathrmbf{Sch}$}}}
\put(120,60){\makebox(0,0){\footnotesize{$\mathrmbf{Oper}$}}}
\put(60,0){\makebox(0,0){\footnotesize{$\mathrmbf{Set}$}}}
\put(-20,50){\oval(20,20)[l]}
\put(-10,50){\oval(20,20)[br]}
\qbezier(-20,40)(-15,40)(-10,40)
\put(0,55){\vector(0,1){0}}
\put(-15,30){\makebox(0,0){\scriptsize{$\mathrmbfit{fmla}$}}}
\put(100,50){\oval(20,20)[l]}
\put(110,50){\oval(20,20)[br]}
\qbezier(100,40)(105,40)(110,40)
\put(120,55){\vector(0,1){0}}
\put(105,25){\makebox(0,0){\scriptsize{$\mathrmbfit{term}$}}}
\put(20,90){\makebox(0,0)[r]{\scriptsize{$\mathrmbfit{sch}$}}}
\put(100,90){\makebox(0,0)[l]{\scriptsize{$\mathrmbfit{oper}$}}}
\put(70,110){\vector(1,-1){40}}
\put(50,110){\vector(-1,-1){40}}
\put(10,50){\vector(1,-1){40}}
\put(110,50){\vector(-1,-1){40}}
\qbezier(50,100)(55,95)(60,90)
\qbezier(70,100)(65,95)(60,90)
\end{picture}}
\put(56,203.5){\vector(4,-1){128}}
\put(-4,143.5){\vector(4,-1){128}}
\put(-64,83.5){\vector(4,-1){128}}
\put(56,83.5){\line(4,-1){28}}\put(116,68.5){\vector(4,-1){68}}
\put(-4,23.5){\vector(4,-1){128}}
\put(140,195){\makebox(0,0){\scriptsize{$\mathrmbfit{spec}$}}}
\put(80,135){\makebox(0,0){\scriptsize{$\mathrmbfit{lang}$}}}
\put(20,50){\makebox(0,0){\scriptsize{$\mathrmbfit{sch}$}}}
\put(140,50){\makebox(0,0){\scriptsize{$\mathrmbfit{oper}$}}}
\put(-30,-28){\makebox(0,0){\footnotesize{$\underbrace{\rule{60pt}{0pt}}$}}}
\put(-30,-40){\makebox(0,0){\footnotesize{${\text{\slshape{semantics}}}$}}}
\put(150,-28){\makebox(0,0){\footnotesize{$\underbrace{\rule{60pt}{0pt}}$}}}
\put(150,-40){\makebox(0,0){\footnotesize{${\text{\slshape{linguistics/formalism}}}$}}}
\put(-190,110){\makebox(0,0){\shortstack{\footnotesize{\slshape{relational}}\\\footnotesize{\slshape{superstructure}}}}}
\put(-130,110){\makebox(0,0){$\left\{\rule{0pt}{30pt}\right.$}}
\put(285,30){\makebox(0,0){\shortstack{\footnotesize{\slshape{functional}}\\\footnotesize{\slshape{base}}}}}
\put(240,30){\makebox(0,0){$\left.\rule{0pt}{30pt}\right\}$}}
\qbezier(22,183)(32,180.5)(42,178)
\qbezier(42,178)(48,184)(54,190)
\end{picture}
\\ \\
\end{tabular}}
\end{center}
\caption{{\ttfamily FOLE} Fibered Architecture}
\label{fbr:arch}
\end{figure}
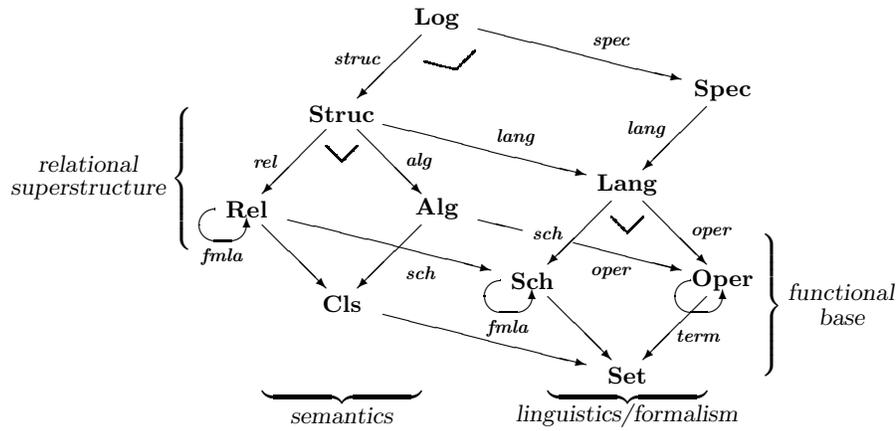
Figure~\ref{fbr:arch} is a 3-dimensional visualization of 
the fibered architecture of the first-order logical environment {\ttfamily FOLE}.
Each node of this figure is a mathematical context,
whereas each edge is a passage between two contexts.
There is a projection
from the 2-D prism below $\mathrmbf{Struc}$ representing the relational superstructure (subsec.~\ref{sec:rel-sup})
  to the 2-D prism below $\mathrmbf{Alg}$   representing the functional base (subsec.~\ref{sec:func-base}). 
The front diamond below $\mathrmbf{Lang}$ represents the linguistics/formalism,
whereas the back diamond below $\mathrmbf{Struc}$ represents the semantics.
The projective passages from semantics to linguistics/formalism
represent the fibration left-to-right and the indexing right-to-left.
The vee-shape at the top of each diamond states that the top mathematical context
is a product of the side contexts modulo the bottom context.
The mathematical contexts on the left side of each diamond form the relational aspect, whereas
the mathematical contexts on the right side form the functional aspect
that lifts the relational to the (first-order) logical aspect.
The 2-D prism below $\mathrmbf{Log}$ represents the institutional architecture.

\section{Components}\label{sec:cmps}

The architectural components (Fig.\ref{fbr:arch}) divide up according to kind and aspect.
The outer level describes the kind of component.
The indexing kind is 
a language (type set, relational schema, operator domain, etc.) (front diamond Fig.\ref{fbr:arch}),
whereas the indexed kind is either a formalism or 
a semantics (classification, relational structure, algebra, etc.) (back diamond Fig.\ref{fbr:arch}).
The inner level describes the aspect of component.
There are basic, relational, functional and logical aspects
(bottom, left, right or top node in either Fig.\ref{fbr:arch} diamond).

\begin{figure}
\begin{center}
\begin{tabular}{c@{\hspace{60pt}}c}
\begin{tabular}{c}
\setlength{\unitlength}{0.9pt}
\begin{picture}(150,120)(0,-10)
\put(75,90){\makebox(0,0){\normalsize{$\top$}}}
\put(30,45){\makebox(0,0){\tiny{physical}}}
\put(120,50){\makebox(0,0){\tiny{abstract}}}
\put(0,0){\makebox(0,0){\tiny{actuality}}}
\put(30,5){\makebox(0,0){\tiny{form}}}
\put(60,0){\makebox(0,0){\tiny{prehension}}}
\put(90,5){\makebox(0,0){\tiny{proposition}}}
\put(120,0){\makebox(0,0){\tiny{nexus}}}
\put(150,5){\makebox(0,0){\tiny{intention}}}
\put(15,60){\makebox(0,0){\tiny{independent}}}
\put(75,60){\makebox(0,0){\tiny{relative}}}
\put(135,60){\makebox(0,0){\tiny{mediating}}}
\put(0,0){\line(2,3){30}}
\put(60,0){\line(-2,3){30}}
\put(120,0){\line(-2,1){90}}
\put(30,5){\line(2,1){90}}
\put(90,5){\line(2,3){30}}
\put(150,5){\line(-2,3){30}}
\multiput(15,60)(2,1){30}{\scriptsize{$\cdot$}}
\multiput(30,45)(1.5,1.5){30}{\scriptsize{$\cdot$}}
\multiput(75,60)(0,2){15}{\scriptsize{$\cdot$}}
\multiput(120,50)(-1.5,1.3){30}{\scriptsize{$\cdot$}}
\multiput(135,60)(-2,1){30}{\scriptsize{$\cdot$}}
\qbezier[35](0,0)(7.5,30)(15,60)
\qbezier[35](30,5)(22.5,30)(15,60)
\qbezier[35](60,0)(67.5,30)(75,60)
\qbezier[35](90,5)(82.5,27.5)(75,60)
\qbezier[35](120,0)(127.5,37.5)(135,60)
\qbezier[35](150,5)(142.5,35)(135,60)
\end{picture}
\end{tabular}
&
\begin{tabular}{c}
\setlength{\unitlength}{0.9pt}
\begin{picture}(150,120)(0,-10)
\put(75,90){\makebox(0,0){\normalsize{$\mathcal{M}$}}}
\put(30,45){\makebox(0,0){\tiny{instance}}}
\put(120,50){\makebox(0,0){\tiny{type}}}
\put(0,-5){\makebox(0,0){\shortstack{\tiny{entity}\\\tiny{instance}\\\scriptsize{$Y$}}}}
\put(30,0){\makebox(0,0){\shortstack{\tiny{entity}\\\tiny{type}\\\scriptsize{$X$}}}}
\put(60,-12){\makebox(0,0){\shortstack{\scriptsize{$\mathrmbf{List}(Y)$}\\\tiny{tuple}\\\mbox{}}}}
\put(93,-10.5){\makebox(0,0){\footnotesize{$\xleftarrow[\tau]{}$}}}
\put(120,-12){\makebox(0,0){\shortstack{\scriptsize{$K$}\\\tiny{key}\\\mbox{}}}}
\put(90,8.5){\makebox(0,0){\shortstack{\mbox{}\\\tiny{signature}\\\scriptsize{$\mathrmbf{List}(X)$}}}}
\put(123.5,5){\makebox(0,0){\footnotesize{$\xleftarrow{\sigma}$}}}
\put(150,8.5){\makebox(0,0){\shortstack{\tiny{relation}\\\tiny{type}\\\scriptsize{$R$}}}}
\put(15,60){\makebox(0,0){\footnotesize{$\mathcal{E}$}}}
\put(75,60){\makebox(0,0){\footnotesize{$\mathrmbf{List}(\mathcal{E})$}}}
\put(111.5,65){\makebox(0,0){\footnotesize{$\stackrel{{\langle{\sigma,\tau}\rangle}}{\leftleftarrows}$}}}
\put(135,60){\makebox(0,0){\footnotesize{$\mathcal{R}$}}}
\put(0,0){\line(2,3){30}}
\put(60,0){\line(-2,3){30}}
\put(120,0){\line(-2,1){90}}
\put(30,5){\line(2,1){90}}
\put(90,5){\line(2,3){30}}
\put(150,5){\line(-2,3){30}}
\multiput(15,60)(2,1){30}{\scriptsize{$\cdot$}}
\multiput(30,45)(1.5,1.5){30}{\scriptsize{$\cdot$}}
\multiput(75,60)(0,2){15}{\scriptsize{$\cdot$}}
\multiput(120,50)(-1.5,1.3){30}{\scriptsize{$\cdot$}}
\multiput(135,60)(-2,1){30}{\scriptsize{$\cdot$}}
\qbezier[35](0,0)(7.5,30)(15,60)
\qbezier[35](30,5)(22.5,30)(15,60)
\qbezier[35](60,0)(67.5,30)(75,60)
\qbezier[35](90,5)(82.5,27.5)(75,60)
\qbezier[35](120,0)(127.5,37.5)(135,60)
\qbezier[35](150,5)(142.5,35)(135,60)
\end{picture}
\end{tabular}
\\ & \\ 
{\footnotesize{top-level categories}}
&
{\footnotesize{{\ttfamily FOLE} components}}
\end{tabular}
\end{center}
\caption{Analogy}
\label{analogy}
\end{figure}
Fig.\ref{analogy} illustrates an analogy between 
the top-level ontological categories discussed in (Sowa~\cite{sowa:kr}) and 
the components 
of the first-order logical environment {\ttfamily FOLE}
(the relational aspect or 2-D prism below $\mathrmbf{Rel}$).
The pair `physical-abstract', 
which corresponds to the Heraclitus distinction {\itshape physis}-{\itshape logos},
is represented in the {\ttfamily FOLE} by a classification
between instances and types of various kinds.
The triples (triads) `actuality-prehension-nexus' and `form-proposition-intention'
correspond to Whitehead's categories of existence.
The latter triple, 
which is analogous to the `entity\mbox{ }type-signature-relation\mbox{ }type' triple,
is represented in the {\ttfamily FOLE} by 
a relational language (schema)
$\mathcal{S} = {\langle{R,\sigma,X}\rangle}$
(Appendix~\ref{sec:ling:fml}).
The former
triple, 
which is analogous to the `entity\mbox{ }instance-tuple-relation\mbox{ }instance' triple,
is represented in the {\ttfamily FOLE} by 
the tuple function
$K\xrightarrow{\tau}\mathrmbf{List}(Y)$
(part of a {\ttfamily FOLE} structure).
The firstness category of `independent(actuality,form)'
is represented in the {\ttfamily FOLE} by
an entity classification 
$\mathcal{E} = {\langle{X,Y,\models_{\mathcal{E}}}\rangle}$
(Appendix~\ref{sec:sem}).
The thirdness category of `mediating(nexus,intention)'
is represented in the {\ttfamily FOLE} by 
a relation classification
$\mathcal{R} = {\langle{R,K,\models_{\mathcal{R}}}\rangle}$
between relational instances (keys) and relational types
(or a classification
between relational instances and logical formula, more generally)
(Appendix~\ref{sec:sem}).
The secondness category of `relative(prehension,proposition)'
is represented in the {\ttfamily FOLE} by 
the list construction of an entity classification 
$\mathrmbf{List}(\mathcal{E}) 
= {\langle{\mathrmbf{List}(X),\mathrmbf{List}(Y),\models_{\mathrmbf{List}(\mathcal{E})}}\rangle}$
between tuples and signatures
(Appendix~\ref{sec:sem}).
Finally,
the entire graph of
the top-level ontological categories
is represented in the {\ttfamily FOLE} by 
a (model-theoretic) structure (classification form)
$\mathcal{M} = {\langle{\mathcal{R},{\langle{\sigma,\tau}\rangle},\mathcal{E}}\rangle}$,
where the relation $\mathcal{R}$ and entity $\mathcal{E}$ classifications are connected by 
a list designation 
${\langle{\sigma,\tau}\rangle} : \mathcal{R} \rightrightarrows \mathrmbf{List}(\mathcal{E})$
(Appendix~\ref{sec:sem}).
This is appropriate,
since a (model-theoretic) structure represents the knowledge in the local world of a community of discourse.

\section{Logical Environment}\label{sec:log-env}
%
The $\mathtt{FOLE}$ institution (logical system) (Kent~\cite{kent:iccs2009})
has at its core the mathematical context of first-order logic (FOL) languages $\mathrmbf{Lang}$.
For any language $\mathcal{L} = {\langle{\mathcal{S},\mathcal{O}}\rangle}$, 
there is a set of constraints $\mathrmbfit{fmla}(\mathcal{L})$
representing the formalism at location $\mathcal{L}$, and
there is a mathematical context of structures $\mathrmbfit{struc}(\mathcal{L})$
representing the semantics at location $\mathcal{L}$.  
For any first-order logic (FOL) language morphism
$\mathcal{L}_{2} = {\langle{\mathcal{S}_{2},\mathcal{O}_{2}}\rangle}
\xrightarrow{\langle{r,f,\omega}\rangle}
{\langle{\mathcal{S}_{1},\mathcal{O}_{1}}\rangle} = \mathcal{L}_{1}$, 
there is a constraint function 
$\mathrmbfit{fmla}(\mathcal{L}_{2})\xrightarrow{\mathrmbfit{fmla}(r,f,\omega)}\mathrmbfit{fmla}(\mathcal{L}_{1})$
(Appendix~\ref{sec:ling:fml})
representing flow of formalism in the forward direction, and
there is a structure passage
$\mathrmbfit{struc}(\mathcal{L}_{2})\xleftarrow{\mathrmbfit{struc}(r,f,\omega)}\mathrmbfit{struc}(\mathcal{L}_{1})$
(Appendix~\ref{sec:sem})
representing flow of semantics in the reverse direction.
This structure passage has a relational component
$\mathrmbf{Rel}(\mathcal{S}_{2}) \xleftarrow{\mathrmbfit{rel}_{{\langle{r,f}\rangle}}} \mathrmbf{Rel}(\mathcal{S}_{2})$
and a functional (algebraic) component
$\mathrmbf{Alg}(\mathcal{O}_{2})\xleftarrow{\mathrmbfit{alg}_{{\langle{f,\omega}\rangle}}}
\mathrmbf{Alg}(\mathcal{O}_{1})$.

$\mathtt{FOLE}$ is an institution, 
since the satisfaction relation is preserved during information flow
along
any first-order logic (FOL) language morphism
$\mathcal{L}_{2} = {\langle{\mathcal{S}_{2},\mathcal{O}_{2}}\rangle}
\xrightarrow{\langle{r,f,\omega}\rangle}
{\langle{\mathcal{S}_{1},\mathcal{O}_{1}}\rangle} = \mathcal{L}_{1}$:
$
\mathrmbfit{struc}(r,f,\omega)(\mathcal{M}_{1}) 
\models_{\mathcal{L}_{2}} ({\langle{I_{2}',s_{2}',\varphi_{2}'}\rangle}\xrightarrow{h_{2}}{\langle{I_{2},s_{2},\varphi_{2}}\rangle})$
\underline{iff}
$\mathcal{M}_{1} \models_{\mathcal{L}_{1}} 
\mathrmbfit{fmla}({\langle{I_{2}',s_{2}',\varphi_{2}'}\rangle}\xrightarrow{h_{2}}{\langle{I_{2},s_{2},\varphi_{2}}\rangle}).
$
In short,
``satisfaction is invariant under change of notation''.
The institution $\mathtt{FOLE}$ is a logical environment,
since
for any language $\mathcal{L} = {\langle{\mathcal{S},\mathcal{O}}\rangle} = {\langle{R,\sigma,X,\Omega}\rangle}$,
if
$\mathcal{M}_{2}
\xrightarrow{{\langle{k,g,h}\rangle}}
\mathcal{M}_{1}
$
is a $\mathrmbfit{lang}$-vertical structure morphism over $\mathcal{L}$, 
then we have the intent order
$\mathcal{M}_{2} \geq_{\mathcal{L}} \mathcal{M}_{1}$;
that is,
$\mathcal{M}_{2} \models_{\mathcal{L}} (\varphi{\;\vdash\;}\psi)$ 
implies 
$\mathcal{M}_{1} \models_{\mathcal{L}} (\varphi{\;\vdash\;}\psi)$ 
for any $\mathcal{S}$-sequent $(\varphi{\;\vdash\;}\psi)$.
In short,
``satisfaction respects structure morphisms''.
(See Appendix~\ref{append:log-env} for a proof of this in the relational aspect.)

\section{Information Systems}\label{sec:inf-sys}
%
Following the theory of general systems,
an information system consists of 
a collection of interconnected parts called information resources and
a collection of part-part relationships between pairs of information resources called constraints.
Semantic information systems have logics
\footnote{A first-order logic 
$\mathcal{L}={\langle{\mathcal{M},\mathcal{T}}\rangle}$
in {\ttfamily FOLE}
consists of a first-order structure 
$\mathcal{M}$
and a first-order specification 
$\mathcal{T}$
that share a common first-order language
$\mathrmbfit{lang}(\mathcal{M})
= \mathrmbfit{lang}(\mathcal{T})$.
A logic enriches
a first-order structure 
with a specification.
The logic is sound when
the structure $\mathcal{M}$
satisfies 
every constraint 
in the specification $\mathcal{T}$.}
as their information resources. 
Just as every logic has an underlying structure, 
so also every information system has an underlying distributed system. 
As such, distributed systems have structures for their component parts.

A {\ttfamily FOLE} distributed system is a passage 
$\mathcal{M} : \mathrmbf{I} \rightarrow \mathrmbf{Struc}$
pictured as a diagram of shape $\mathrmbf{I}$ 
within the ambient mathematical context of first-order structures.
As such,
it consists of an indexed family 
$\{ \mathcal{M}_{i}
\mid i \in |\mathrmbf{I}| \}$
of structures together with an indexed family 
$\{ \mathcal{M}_{i} 
\xrightarrow{m_{e}}
\mathcal{M}_{j} 
\mid (e : i \rightarrow j) \in \mathrmbf{I} \}$
of structure morphisms.
A {\ttfamily FOLE} (semantic) information system is a diagram
$\mathcal{L} : \mathrmbf{I} \rightarrow \mathrmbf{Log}$ 
within the mathematical context of first-order logics.
This consists of an indexed family of logics
$\{ \mathcal{L}_{i}
: i \in |\mathrmbf{I}| \}$
and an indexed family of logic morphisms
$\{ \mathcal{L}_{i} \xrightarrow{l_{e}} \mathcal{L}_{j} 
\mid (e : i \rightarrow j) \in \mathrmbf{I} \}$.
An information system $\mathcal{L}$ 
has an underlying distributed system
$\mathcal{M} = \mathcal{L}\circ\mathrmbfit{struc}$
of the same shape with $\mathcal{M}_{i} = \mathrmbfit{struc}(\mathcal{L}_{i})$ for all $i \in |\mathrmbf{I}|$.
An information channel 
${\langle{\gamma:\mathcal{M}\Rightarrow\Delta(\mathcal{C}),\mathcal{C}}\rangle}$
consists of an indexed family
$\{ \mathcal{M}_{i}\xrightarrow{\gamma_{i}}\mathcal{C} \mid i \in |\mathrmbf{I}| \}$
of structure morphisms with a common target structure $\mathcal{C}$
called the core of the channel.
Information flows along channels.
We are mainly interested in channels
that cover a distributed system $\mathcal{M} : \mathrmbf{I} \rightarrow \mathrmbf{Struc}$,
where the part-whole relationships respect the system constraints (are consistent with
the part-part relationships).
In this case, 
there exist optimal channels.
An optimal core is called the sum of the distributed system,
and the optimal channel components (structure morphisms) are flow links.

System interoperability is defined by moving formalism over semantics.
The fusion (unification) $\coprod\mathcal{L}$ of the information system $\mathcal{L}$
represents the whole system in a centralized fashion.
The fusion logic is defined by direct system flow:
(i) direct logic flow of the component parts of the information system 
along the optimal channel over the underlying distributed system to a centralized location 
(the mathematical context of structures at the optimal channel core), and 
(ii) product combining the contributions of the parts into a whole.
The consequence $\mathcal{L}^{\scriptscriptstyle\blacklozenge}$ of the information system $\mathcal{L}$
represents the whole system in a distributed fashion.
This is an information system defined by inverse system flow: 
(i) consequence of the fusion logic, and 
(ii) inverse logic flow of this consequence back along the same optimal channel,
transfering the constraints of the whole system (the fusion logic)
to the distributed locations (structures) of the component parts.
See Kent~\cite{kent:iccs2009} for further details.
\footnote{In light of the transformation described in Appendix~\ref{rel:log},
an information system of sound logics
can be regarded as a system of logical/relational databases.
The system consequence of such systems represents database interoperabilty.
Kent~\cite{kent:iccs2009} has more details about the information flow of sound logics
in an arbitrary logical environment.}
%

\section{Summary and Future Work}\label{sec:sum-fut-wrk}

In this paper we have described the first-order logical environment {\ttfamily FOLE} in classification form.
This gives a holistic treatment of first-order logic,
by the use of several novel elements:
the use of signatures (type lists) for relational arities, in place of ordinal numbers;
the use of abstract tuples (relational instances, keys), 
thus making {\ttfamily FOLE} 
compatible with relational databases;
the use of classifications
for both entities and relations; and
the use of relational constraints for the sentences of the {\ttfamily FOLE} institution.
{\ttfamily FOLE} also has an interpretation form (Kent~\cite{kent:db:sem}) 
that represents the formalism and semantics of logical/relational databases, including relational algebra.
There are transformational passages between the classification form and a strict version of the interpretation form.
Appendix~\ref{rel:log} 
briefly discusses
the transformation from sound logics to logical/relational databases.

{\ttfamily FOLE} has advantages over other approaches to first-order logic:
in {\ttfamily FOLE} the formalism is completely integrated into the semantics;
the classification form of {\ttfamily FOLE} has a natural extension to relational/logical databases,
as represented by the interpretation form of {\ttfamily FOLE}; and
{\ttfamily FOLE} is a logical environment,
thus allowing practitioners a rigorously defined approach 
towards the interoperation of online semantic systems of information resources 
that include relational databases.

Future work includes:
finishing work on the interpretation form of {\ttfamily FOLE};
further work on defining the transformational passages between the classification and interpretation forms; 
developing a linearization process from {\ttfamily FOLE} to sketch-like forms of logic such as Ologs 
(Spivak and Kent~\cite{spivak:kent:olog}); and
linking {\ttfamily FOLE} with the Common Logic standard.

\appendix

\newpage
\section{Appendix}\label{sec:append}


\subsection{Functional Base.}\label{sec:func-base}

\comment{
\begin{center}
{\begin{tabular}{c}
\setlength{\unitlength}{0.45pt}
\begin{picture}(360,10)(-185,30)
\put(-90,30){\begin{picture}(0,0)(0,0)
\put(120,60){\makebox(0,0){\scriptsize{$\mathrmbf{Alg}$}}}
\put(60,0){\makebox(0,0){\scriptsize{$\mathrmbf{Cls}$}}}
\put(80,35){\makebox(0,0){\scriptsize{$\mathrmbfit{cls}$}}}
\put(110,50){\vector(-1,-1){40}}
\end{picture}}
\put(90,-15){\begin{picture}(0,0)(0,0)
\put(120,60){\makebox(0,0){\scriptsize{$\mathrmbf{Oper}$}}}
\put(60,0){\makebox(0,0){\scriptsize{$\mathrmbf{Set}$}}}
\put(100,50){\oval(20,20)[l]}
\put(110,50){\oval(20,20)[br]}
\qbezier(100,40)(105,40)(110,40)
\put(120,55){\vector(0,1){0}}
\put(115,25){\makebox(0,0){\scriptsize{$\mathrmbfit{set}$}}}
\put(60,45){\makebox(0,0){\scriptsize{$\mathrmbfit{term}$}}}
\put(110,50){\vector(-1,-1){40}}
\end{picture}}
\put(56,83.5){\vector(4,-1){128}}
\put(-4,23.5){\vector(4,-1){128}}
\put(120,80){\makebox(0,0){\scriptsize{$\mathrmbfit{oper}$}}}
\put(60,20){\makebox(0,0){\scriptsize{$\mathrmbfit{set}$}}}
\put(295,30){\makebox(0,0){\shortstack{\scriptsize{\slshape{functional}}\\\scriptsize{\slshape{base}}}}}
\put(240,30){\makebox(0,0){$\left.\rule{0pt}{30pt}\right\}$}}
\end{picture}
\end{tabular}}
\end{center}
}

\subsubsection{Linguistics/Formalism.}

\paragraph{\underline{Base Linguistics}: $\mathrmbf{Set}$.}
%
\begin{picture}(0,0)(0,0)
\put(0,0){
\setlength{\unitlength}{0.45pt}
\begin{picture}(360,10)(-185,-60)
\put(-90,30){\begin{picture}(0,0)(0,0)
\put(120,60){\makebox(0,0){\scriptsize{$\mathrmbf{Alg}$}}}
\put(60,0){\makebox(0,0){\scriptsize{$\mathrmbf{Cls}$}}}
\put(76,35){\makebox(0,0){\scriptsize{$\mathrmbfit{cls}$}}}
\put(110,50){\vector(-1,-1){40}}
\end{picture}}
\put(90,-15){\begin{picture}(0,0)(0,0)
\put(120,60){\makebox(0,0){\scriptsize{$\mathrmbf{Oper}$}}}
\put(60,0){\makebox(0,0){\scriptsize{$\mathrmbf{Set}$}}}
\put(100,50){\oval(20,20)[l]}
\put(110,50){\oval(20,20)[br]}
\qbezier(100,40)(105,40)(110,40)
\put(120,55){\vector(0,1){0}}
\put(106,27){\makebox(0,0){\scriptsize{$\mathrmbfit{set}$}}}
\put(68,52){\makebox(0,0){\scriptsize{$\mathrmbfit{term}$}}}
\put(110,50){\vector(-1,-1){40}}
\end{picture}}
\put(56,83.5){\vector(4,-1){128}}
\put(-4,23.5){\vector(4,-1){128}}
\put(120,80){\makebox(0,0){\scriptsize{$\mathrmbfit{oper}$}}}
\put(60,20){\makebox(0,0){\scriptsize{$\mathrmbfit{set}$}}}
\put(295,30){\makebox(0,0){\shortstack{\scriptsize{\slshape{functional}}\\\scriptsize{\slshape{base}}}}}
\put(240,30){\makebox(0,0){$\left.\rule{0pt}{30pt}\right\}$}}
\end{picture}
}
\end{picture}
%
%
A set (of entity types) $X$
defines
a mathematical context of type lists (signatures)
$\mathrmbf{List}(X) = (\mathrmbf{Set}{\downarrow}X)$.
The {\ttfamily FOLE} uses type lists for relational arities, instead of ordinal numbers.

The first subcomponent of any linguistic component is a set of entity types (sorts) $X$.
Examples of entity types are
`human' representing the set of all human beings,
`blue' representing the set of all objects of color blue, etc. 
A type list (signature) ${\langle{I,s}\rangle}$ 
consists of an arity set $I$ and a type map $I \xrightarrow{s} X$ mapping elements of the arity to entity types.
This can be denoted by the list notation
$(\ldots{s_{i}}\ldots)$
or the type declaration notation
$(\ldots{i{\,:}s_{i}}\ldots)$
for
$i{\,\in\,}I$ and $s_{i}{\,\in\,}X$.
For example,
the type list
`{\footnotesize{(make:{\bfseries String},model:{\bfseries String},year:{\bfseries Number},color:{\bfseries Color})}}'
is a type list for cars with 
valence 4,
arity set {\footnotesize{$\{\text{make},\text{model},\text{year},\text{color}\}$}}, and
type map
{\footnotesize{$\{\text{make}\mapsto\text{{\bfseries String}},\cdots\}$}}.
A type list morphism 
${\langle{I_{2},s_{2}}\rangle} \xrightarrow{h} {\langle{I_{1},s_{1}}\rangle}$
is an arity function $I_{2} \xrightarrow{h} I_{1}$
that satisfies the commutative diagram $h \cdot s_{1} = s_{2}$.
We say that 
$s_{2}$ is at least as general as $s_{1}$.

Given the natural numbers $\aleph = \{0,1,\cdots \}$,
let $\underline{\aleph}$ denote 
the mathematical context of finite ordinals (number sets) $\underline{n}=\{0,1,\cdots,n{-}1\}$
and functions between them.
This is the skeleton of
the mathematical context $\mathrmbf{Fin}$ of finite sets and functions.
Both represent the single-sorted case where $X = \mathbf{1}$.
We have the following inclusion of base language mathematical contexts.
\footnote{We use the mathematical context 
$\overset{\scriptscriptstyle\ast}{\mathrmbf{List}}(X) = (\mathrmbf{Fin}{\downarrow}X)$
for type lists of finite arity.}
\[
\underset{\text{\shortstack{\rule{0pt}{10pt}skeleton}}}
{\underline{\aleph}}
\!\!\subseteq\;\; 
\underset{\text{\shortstack{\rule[2pt]{0pt}{10pt}single-sorted}}}
{\mathrmbf{Fin}}
\;\;\subseteq\;\; 
\underset{\text{\rule{0pt}{10pt}many-sorted}}
{\overset{\scriptscriptstyle\ast}{\mathrmbf{List}}(X)}
\]
Traditional first-order systems use the natural numbers $\aleph$ for indexing relations.
More flexible first-order systems, such as {\ttfamily FOLE} or relational database systems, 
use finite sets when single-sorted or type lists when many-sorted.

\paragraph{\underline{Algebraic Linguistics}: $\mathrmbf{Oper}\xrightarrow{\mathrmbfit{set}}\mathrmbf{Set}$.}


A functional language (operator domain) is a pair ${\langle{X,\Omega}\rangle}$,
where $X$ is a set of entity types (sorts) and $\Omega$ is an $X$-sorted operator domain;
that is,
$\Omega = \{ \Omega_{x,{\langle{I,s}\rangle}} \mid x \in X, {\langle{I,s}\rangle} \in \overset{\scriptscriptstyle\ast}{\mathrmbf{List}}(X) \}$
is a collection of sets of function (operator) symbols,
where $e \in \Omega_{x,{\langle{I,s}\rangle}}$ is 
a function symbol of entity type (sort) $x$
and finite arity ${\langle{I,s}\rangle}$,
\footnote{This is a slight misnomer,
since ${\langle{I,s}\rangle}$ is actually the signature of the function symbol.
whereas the arity of $e$ is the indexing set $I$
and the valence of $e$ is the cardinality $|I|$.}
symbolized by $x \xrightharpoondown{e} {\langle{I,s}\rangle}$.
An element $c \in {\Omega}_{x,{\langle{\emptyset,0_{X}}\rangle}}$ is called a constant symbol of sort $x$.
%
%
Any operator domain ${\langle{X,\Omega}\rangle}$
defines a mathematical context of terms
$\mathrmbf{Term}_{{\langle{X,\Omega}\rangle}}$,
whose objects are $X$-signatures ${\langle{I,s}\rangle}$ and 
whose morphisms are term vectors ${\langle{I',s'}\rangle} \xrightharpoondown{t} {\langle{I,s}\rangle}$,
where 
$t = \{ {s'}_{i'} \xrightharpoondown{t_{i'}} {\langle{I,s}\rangle} \mid i' \in I' \}$
is an indexed collection (vector) of ${\langle{I,s}\rangle}$-ary terms.
Terms and term vectors are defined by mutual induction.

%
A morphism of functional languages 
is a pair
${\langle{X_{2},\Omega_{2}}\rangle} \xrightarrow{{\langle{f,\omega}\rangle}} {\langle{X_{1},\Omega_{1}}\rangle}$, 
where $X_{2} \xrightarrow{f} X_{1}$ is a function of entity types (sorts)
and
$\omega : \Omega_{2} \rightarrow \Omega_{1}$
is a collection
$\{ 
{(\Omega_{2})}_{x_{2},{\langle{I_{2},s_{2}}\rangle}} \xrightarrow{\omega_{x_{2},{\langle{I_{2},s_{2}}\rangle}}}
{(\Omega_{1})}_{f(x_{2}),{{\scriptscriptstyle\sum}_{f}(I_{2},s_{2})}}
\mid x_{2} \in X_{2}, {\langle{I_{2},s_{2}}\rangle} \in \overset{\scriptscriptstyle\ast}{\mathrmbf{List}}(X_{2}) \}$ 
of maps between function symbol sets:
$\omega$ maps 
a function symbol 
$x_{2} \xrightharpoondown{e} {\langle{I_{2},s_{2}}\rangle}$
in $\Omega_{2}$ to 
a function symbol 
$f(x_{2}) \xrightharpoondown{\omega(e)} 
{\scriptstyle\sum}_{f}(I_{2},s_{2})
={\langle{I_{2},s_{2}{\,\cdot\,}f}\rangle}$
in $\Omega_{1}$.
Given any morphism of functional languages
${\langle{X_{2},\Omega_{2}}\rangle} \xrightarrow{{\langle{f,\omega}\rangle}} {\langle{X_{1},\Omega_{1}}\rangle}$,
there is a term passage 
$\mathrmbf{Term}_{{\langle{X_{2},\Omega_{2}}\rangle}}
\xrightarrow{\mathrmbfit{term}_{{\langle{f,\omega}\rangle}}}
\mathrmbf{Term}_{{\langle{X_{1},\Omega_{1}}\rangle}}$ defined by
induction.
\comment{
\begin{center}
{\footnotesize{$\begin{array}{r@{\hspace{8pt}=\hspace{8pt}}l}
\mathrmbfit{term}_{{\langle{f,\omega}\rangle}}(I_{2},s_{2})
&
{\langle{I_{2},s_{2}{\,\cdot\,}f}\rangle}
\\
\mathrmbfit{term}_{{\langle{f,\omega}\rangle}}({\langle{I_{2},s_{2}}\rangle} \xrightharpoondown{t_{2}} {\langle{I'_{2},s'_{2}}\rangle})
&
\left({\langle{I_{2},s_{2}{\,\cdot\,}f}\rangle} \xrightharpoondown{\omega^{\ast}(t_{2})}{\langle{I'_{2},s'_{2}{\,\circ\,}f}\rangle}\right)
\end{array}$}}
\end{center}
}
Let $\mathrmbf{Oper}$ denote the mathematical context of functional languages (operator domains).

\paragraph{\underline{Algebraic Formalism}.}


Let $\mathcal{O} = {\langle{X,\Omega}\rangle}$ be an operator domain.
%
%
An $\mathcal{O}$-equation 
is 
a parallel pair of term vectors 
${\langle{I',s'}\rangle} \xrightharpoondown{t,t'} {\langle{I,s}\rangle}$.
We represent an equation using the traditional notation $(t{\,=\,}t')$.
%
%
An equational presentation ${\langle{X,\Omega,E}\rangle}$ 
consists of an operator domain $\mathcal{O} = {\langle{X,\Omega}\rangle}$ 
and a set of $\mathcal{O}$-equations $E$.
A congruence is any equational presentation closed under left and right term composition.
Any equational presentation ${\langle{X,\Omega,E}\rangle}$ 
generates a congruence ${\langle{X,\Omega,E^{\scriptscriptstyle\bullet}}\rangle}$,
which defines a quotient mathematical context of terms
$\mathrmbf{Term}_{{\langle{X,\Omega,E}\rangle}}$
with a morphism 
${\langle{I',s'}\rangle} \xrightharpoondown{[t]} {\langle{I,s}\rangle}$
being an equivalence class of terms.
There is a canonical passage
$\mathrmbf{Term}_{{\langle{X,\Omega}\rangle}}\xrightarrow{[]}\mathrmbf{Term}_{{\langle{X,\Omega,E}\rangle}}$.
A morphism of equational presentations
${\langle{X_{2},\Omega_{2},E_{2}}\rangle} \xrightarrow{{\langle{f,\omega}\rangle}} {\langle{X_{1},\Omega_{1},E_{1}}\rangle}$
is a morphism of functional languages
${\langle{X_{2},\Omega_{2}}\rangle} \xrightarrow{{\langle{f,\omega}\rangle}} {\langle{X_{1},\Omega_{1}}\rangle}$
that preserves equations:
an $\mathcal{O}_{2}$-equation  
${\langle{I_{2}',s_{2}'}\rangle} \xrightharpoondown{t_{2},t_{2}'} {\langle{I_{2},s_{2}}\rangle}$
in $E_{2}$
is mapped 
to an $\mathcal{O}_{1}$-equation  
${{\scriptstyle\sum}_{f}(I_{2}',s_{2}')}
\xrightharpoondown{\omega^{\ast}(t),\omega^{\ast}(t')} 
{{\scriptstyle\sum}_{f}(I_{2},s_{2})}$
in the congruence $E_{1}^{\scriptscriptstyle\bullet}$.
Hence,
there is a term passage
$\mathrmbf{Term}_{{\langle{X_{2},\Omega_{2},E_{2}}\rangle}}
\xrightarrow{\mathrmbfit{term}_{{\langle{f,\omega}\rangle}}}
\mathrmbf{Term}_{{\langle{X_{1},\Omega_{1},E_{1}}\rangle}}$
that commutes with canons.

\subsubsection{Semantics.}

\paragraph{\underline{Base Semantics}: $\mathrmbf{Cls}\xrightarrow{\mathrmbfit{typ}}\mathrmbf{Set}$.}


For any entity classification $\mathcal{E} = {\langle{X,Y,\models_{\mathcal{E}}}\rangle}$, 
there is a tuple passage
$\mathrmbf{List}(X)^{\mathrm{op}}\xrightarrow{\mathrmbfit{tup}_{\mathcal{E}}}\mathrmbf{Set}$
defined as the extent of the list classification $\mathrmbf{List}(\mathcal{E})$.
It maps a type list (signature) ${\langle{I,s}\rangle}\in\mathrmbf{List}(X)$ 
to its extent
$\mathrmbfit{tup}_{\mathcal{E}}(I,s) = \mathrmbfit{ext}_{\mathrmbf{List}(\mathcal{E})}(I,s) \subseteq \mathrmbf{List}(Y)$.
%
%
An entity infomorphism
${\langle{f,g}\rangle} : \mathcal{E}_{2} \rightleftarrows \mathcal{E}_{1}$
defines a bridge 
$\mathrmbfit{tup}_{\mathcal{E}_{2}}\stackrel{\tau_{{\langle{f,g}\rangle}}}{\Longleftarrow}
({\scriptstyle\sum}_{f})^{\mathrm{op}}\circ\mathrmbfit{tup}_{\mathcal{E}_{1}}$
between tuple passages.
For any source signature ${\langle{I_{2},s_{2}}\rangle} \in (\mathrmbf{Set}{\downarrow}X_{2})$,
the tuple function
$\tau_{{\langle{f,g}\rangle}}(I_{2},s_{2}) = {(\mbox{-})} \cdot g : 
\mathrmbfit{tup}_{\mathcal{E}_{1}}({\scriptstyle\sum}_{f}(I_{2},s_{2})) \rightarrow \mathrmbfit{tup}_{\mathcal{E}_{2}}(I_{2},s_{2})$
is define by composition.

\comment{
\begin{center}
\begin{tabular}{c}
\setlength{\unitlength}{0.6pt}
\begin{picture}(120,90)(0,7)
\put(0,80){\makebox(0,0){\footnotesize{$\mathrmbf{List}(X_{2})^{\mathrm{op}}$}}}
\put(124,80){\makebox(0,0){\footnotesize{$\mathrmbf{List}(X_{1})^{\mathrm{op}}$}}}
\put(60,5){\makebox(0,0){\footnotesize{$\mathrmbf{Set}$}}}
\put(60,92){\makebox(0,0){\scriptsize{$({\scriptstyle\sum}_{f})^{\mathrm{op}}$}}}
\put(24,38){\makebox(0,0)[r]{\scriptsize{$\mathrmbfit{tup}_{\mathcal{E}_{2}}$}}}
\put(97,38){\makebox(0,0)[l]{\scriptsize{$\mathrmbfit{tup}_{\mathcal{E}_{1}}$}}}
\put(60,54){\makebox(0,0){\shortstack{\scriptsize{$\tau_{{\langle{f,g}\rangle}}$}\\\large{$\Longleftarrow$}}}}
\put(40,80){\vector(1,0){40}}
\put(9,68){\vector(3,-4){40}}
\put(111,68){\vector(-3,-4){40}}
\end{picture}
\end{tabular}
\end{center}
}

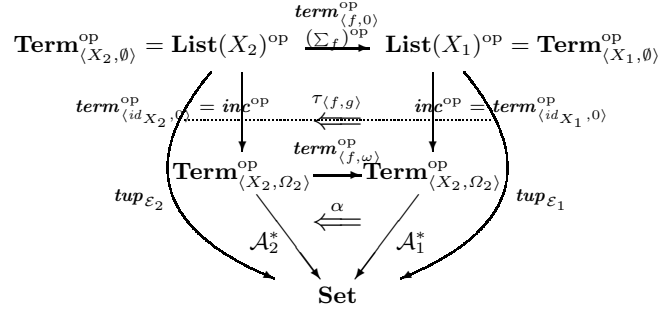
\begin{figure}
\begin{center}
\begin{tabular}{c}
\setlength{\unitlength}{0.6pt}
\begin{picture}(120,180)(0,-75)
\put(30,80){\makebox(0,0)[r]{\footnotesize{$
\mathrmbf{Term}_{{\langle{X_{2},\emptyset}\rangle}}^{\mathrm{op}}=\mathrmbf{List}(X_{2})^{\mathrm{op}}$}}}
\put(90,80){\makebox(0,0)[l]{\footnotesize{$\mathrmbf{List}(X_{1})^{\mathrm{op}}=\mathrmbf{Term}_{{\langle{X_{1},\emptyset}\rangle}}^{\mathrm{op}}$}}}
\put(0,0){\makebox(0,0){\footnotesize{$\mathrmbf{Term}_{{\langle{X_{2},\Omega_{2}}\rangle}}^{\mathrm{op}}$}}}
\put(120,0){\makebox(0,0){\footnotesize{$\mathrmbf{Term}_{{\langle{X_{2},\Omega_{2}}\rangle}}^{\mathrm{op}}$}}}
\put(60,95){\makebox(0,0){\scriptsize{$\overset{\textstyle{\mathrmbfit{term}_{{\langle{f,0}\rangle}}^{\mathrm{op}}}}{({\scriptscriptstyle\sum}_{f})^{\mathrm{op}}}$}}}
\put(60,15){\makebox(0,0){\scriptsize{$\mathrmbfit{term}_{{\langle{f,\omega}\rangle}}^{\mathrm{op}}$}}}
\put(18,40){\makebox(0,0)[r]{\scriptsize{$\mathrmbfit{term}_{{\langle{\mathrmit{id}_{X_{2}},0}\rangle}}^{\mathrm{op}} = \mathrmbfit{inc}^{\mathrm{op}}$}}}
\put(108,40){\makebox(0,0)[l]{\scriptsize{$\mathrmbfit{inc}^{\mathrm{op}}=\mathrmbfit{term}_{{\langle{\mathrmit{id}_{X_{1}},0}\rangle}}^{\mathrm{op}}$}}}
\put(40,80){\vector(1,0){40}}
\put(45,0){\vector(1,0){30}}
\put(0,65){\vector(0,-1){50}}
\put(120,65){\vector(0,-1){50}}
\put(60,-75){\makebox(0,0){\footnotesize{$\mathrmbf{Set}$}}}
\put(24,-42){\makebox(0,0)[r]{\footnotesize{$\mathcal{A}^{\ast}_{2}$}}}
\put(97,-42){\makebox(0,0)[l]{\footnotesize{$\mathcal{A}^{\ast}_{1}$}}}
\put(60,-26){\makebox(0,0){\shortstack{\scriptsize{$\alpha$}\\\large{$\Longleftarrow$}}}}
\put(9,-12){\vector(3,-4){40}}
\put(111,-12){\vector(-3,-4){40}}
\qbezier(-20,65)(-90,-20)(20,-65)\put(20,-65){\vector(2,-1){0}}
\qbezier(140,65)(210,-20)(100,-65)\put(100,-65){\vector(-2,-1){0}}
\put(-48,-15){\makebox(0,0)[r]{\scriptsize{$\mathrmbfit{tup}_{\mathcal{E}_{2}}$}}}
\put(172,-15){\makebox(0,0)[l]{\scriptsize{$\mathrmbfit{tup}_{\mathcal{E}_{1}}$}}}
\put(60,40){\makebox(0,0){\shortstack{\scriptsize{$\tau_{{\langle{f,g}\rangle}}$}\\\large{$\Longleftarrow$}}}}
\qbezier[80](-36,35)(60,35)(156,35)
\end{picture}
\end{tabular}
\end{center}
\caption{Functional Base Interpretation}
\label{fig:func:interp}
\end{figure}
%

\paragraph{\underline{Algebraic Semantics}: 
$\mathrmbf{Cls}\xleftarrow{\mathrmbfit{cls}}\mathrmbf{Alg}\xrightarrow{\mathrmbfit{oper}}\mathrmbf{Oper}$.}


A many-sorted algebra $\mathcal{A} = {\langle{\mathcal{E},\mathcal{O},{\langle{A,\delta}\rangle}}\rangle}$
consists of 
an entity classification 
$\mathcal{E} = {\langle{X,Y,\models_{\mathcal{E}}}\rangle}$,
an operator domain $\mathcal{O} = {\langle{X,\Omega}\rangle}$,
and
an $\mathcal{O}$-algebra ${\langle{A,\delta}\rangle}$
compatible with $\mathcal{E}$,
where $A = \{ A_{x} \mid x \in X \}$ is an $X$-sorted set
and $\delta$ assigns an ${\langle{I,s}\rangle}$-ary $x$-sorted function (operation) 
$A_{x} \xleftarrow{\delta_{e}} A^{{\langle{I,s}\rangle}}$
to each function symbol 
$x \xrightharpoondown{e} {\langle{I,s}\rangle}$
with 
$
A^{{\langle{I,s}\rangle}} = \prod_{i\in{I}} A_{s_{i}}$ the product set.
A many-sorted algebra $\mathcal{A} = {\langle{\mathcal{E},\mathcal{O},{\langle{A,\delta}\rangle}}\rangle}$
defines (by induction) an algebraic interpretation passage
$\mathrmbf{Term}_{{\langle{X,\Omega}\rangle}}^{\mathrm{op}}\xrightarrow{\mathcal{A}^{\ast}}\mathrmbf{Set}$,
which extends the tuple passage
$\mathrmbfit{tup}_{\mathcal{E}} = \mathrmbfit{inc}^{\mathrm{op}}{\;\circ\;}\mathcal{A}^{\ast}$
by compatibility.
An algebra $\mathcal{A}$ satisfies an equation $(t{\,=\,}t')$,
symbolized by $\mathcal{A}{\;\models\;}(t=t')$,
when the interpretation 
maps the terms to the same function $\mathcal{A}^{\ast}(t)=\mathcal{A}^{\ast}(t')$.
%
%
A many-sorted algebraic homomorphism
$\mathcal{A}_{2} = {\langle{\mathcal{E}_{2},\mathcal{O}_{2},{\langle{A_{2},\delta_{2}}\rangle}}\rangle}
	\xrightarrow{\langle{f,g,\omega,h}\rangle}
{\langle{\mathcal{E}_{1},\mathcal{O}_{1},{\langle{A_{1},\delta_{1}}\rangle}}\rangle} = \mathcal{A}_{1}$
consists of 
an entity infomorphism
${\langle{f,g}\rangle} : 
\mathcal{E}_{2}
\rightleftarrows 
\mathcal{E}_{1}$,
a morphism of many-sorted operator domains
${{\langle{f,\omega}\rangle}} : \mathcal{O}_{2} \rightarrow \mathcal{O}_{1}$, and
an
$\mathcal{O}_{2}$-algebra morphism 
${\langle{A_{2},\delta_{2}}\rangle} \xleftarrow{h} \mathrmbfit{alg}_{{\langle{f,\omega}\rangle}}(A_{1},\delta_{1})$
compatible with ${\langle{f,g}\rangle}$.
A many-sorted algebraic homomorphism
$\mathcal{A}_{2}
	\xrightarrow{\langle{f,g,\omega,h}\rangle}
\mathcal{A}_{1}$
defines an algebraic bridge 
$\mathcal{A}^{\ast}_{2}
\stackrel{\alpha}{\Longleftarrow}
{\mathrmbfit{term}_{{\langle{f,\omega}\rangle}}}^{\mathrm{op}}{\;\circ\;}\mathcal{A}^{\ast}_{1}$
between algebraic interpretations,
which extends the tuple bridge
$\tau_{{\langle{f,g}\rangle}} = \mathrmbfit{inc}^{\mathrm{op}}{\;\circ\;}\alpha$
by compatibility.
Let $\mathrmbf{Alg}$ denote the mathematical context of many-sorted algebras.
(The base semantics embeds into the functional semantics Fig.~\ref{fig:func:interp}.)
\vspace{3pt}


\subsection{Relational Superstructure.}\label{sec:rel-sup}

\subsubsection{Linguistics/Formalism.}\label{sec:ling:fml}

\paragraph{\underline{Relational Linguistics}: $\mathrmbf{Sch}$.}

%
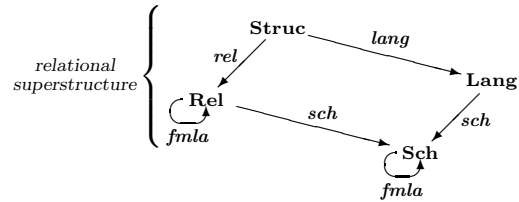
\begin{picture}(0,0)(0,0)
\put(0,0){
\setlength{\unitlength}{0.45pt}
\begin{picture}(360,40)(-330,27)
\put(-90,30){\begin{picture}(0,0)(0,0)
\put(60,120){\makebox(0,0){\scriptsize{$\mathrmbf{Struc}$}}}
\put(0,60){\makebox(0,0){\scriptsize{$\mathrmbf{Rel}$}}}
\put(27,97){\makebox(0,0)[r]{\scriptsize{$\mathrmbfit{rel}$}}}
\put(50,110){\vector(-1,-1){40}}
\put(-20,50){\oval(20,20)[l]}
\put(-10,50){\oval(20,20)[br]}
\qbezier(-20,40)(-15,40)(-10,40)
\put(0,55){\vector(0,1){0}}
\put(-15,28){\makebox(0,0){\scriptsize{$\mathrmbfit{fmla}$}}}
\end{picture}}
\put(90,-15){\begin{picture}(0,0)(0,0)
\put(60,120){\makebox(0,0){\scriptsize{$\mathrmbf{Lang}$}}}
\put(0,60){\makebox(0,0){\scriptsize{$\mathrmbf{Sch}$}}}
\put(60,88){\makebox(0,0)[r]{\scriptsize{$\mathrmbfit{sch}$}}}
\put(50,110){\vector(-1,-1){40}}
\put(-20,50){\oval(20,20)[l]}
\put(-10,50){\oval(20,20)[br]}
\qbezier(-20,40)(-15,40)(-10,40)
\put(0,55){\vector(0,1){0}}
\put(-15,28){\makebox(0,0){\scriptsize{$\mathrmbfit{fmla}$}}}
\end{picture}}
\put(-4,143.5){\vector(4,-1){128}}
\put(-64,83.5){\vector(4,-1){128}}
\put(65,140){\makebox(0,0){\scriptsize{$\mathrmbfit{lang}$}}}
\put(8,79){\makebox(0,0){\scriptsize{$\mathrmbfit{sch}$}}}
\put(-200,110){\makebox(0,0){\shortstack{\scriptsize{\slshape{relational}}\\\scriptsize{\slshape{superstructure}}}}}
\put(-135,110){\makebox(0,0){$\left\{\rule{0pt}{30pt}\right.$}}
\end{picture}}
\end{picture}
%

\paragraph{Schemas.}

A relational language (schema)
$\mathcal{S} = {\langle{R,\sigma,X}\rangle}$
has two components:
a base and a superstructure built upon the base.
The base consists of 
a set of entity types (sorts) $X$,
which defines the type list mathematical context
$\mathrmbf{List}(X)$.
The superstructure
consists of 
a set of relation types (symbols) $R$ and
a (discrete) type list passage $R \xrightarrow{\sigma} \mathrmbf{List}(X)$
mapping a relation symbol $r \in R$ to its type list
$\sigma(r) = {\langle{I,s}\rangle}$.
A relational language (schema) morphism  
$\mathcal{S}_{2} = {\langle{R_{2},\sigma_{2},X_{2}}\rangle} \stackrel{{\langle{r,f}\rangle}}{\Longrightarrow}
{\langle{R_{1},\sigma_{1},X_{1}}\rangle} = \mathcal{S}_{1}$
also has two components:
a base and a superstructure built upon the base.
The base consists of 
an entity type (sort) function $f : X_{2} \rightarrow X_{1}$,
which defines the type list passage
$\mathrmbf{List}(X_{2})\xrightarrow{{\scriptscriptstyle\sum}_{f}}\mathrmbf{List}(X_{1})$
mapping a type list
$(\ldots{s_{i_{2}}}\ldots)$
to the type list
$(\ldots{f(s_{i_{2}})}\ldots)$.
The superstructure consists of 
a relation type function   $r : R_{2} \rightarrow R_{1}$ 
which preserves type lists, 
satisfying the condition
$r \cdot \sigma_{1} = \sigma_{2} \cdot {\scriptstyle\sum}_{f}$. 
Let $\mathrmbf{Sch}$ symbolize the mathematical context of relational languages (schemas)
with type set projection passage $\mathrmbf{Sch}\xrightarrow{\mathrmbfit{set}}\mathrmbf{Set}$.

\paragraph{Formulas.}

For any type list ${\langle{I,s}\rangle}$,
let $R(I,s) \subseteq R$ denote the set of all relation types with this type list.
These are called ${\langle{I,s}\rangle}$-ary relation symbols.
Formulas form a schema 
$\mathrmbfit{fmla}(\mathcal{S}) = {\langle{\widehat{R},\widehat{\sigma},X}\rangle}$ 
that extends $\mathcal{S}$:
with inductive definitions,
the set of relation types 
is extended
to a set of logical formulas $\widehat{R}$ and
the relational type list function 
is extended
to a type list function $\widehat{R} \xrightarrow{\widehat{\sigma}} \mathrmbf{List}(X)$.
For any type list ${\langle{I,s}\rangle}$,
let $\widehat{R}(I,s) \subseteq \widehat{R}$ 
denote the set of all formulas with this type list.
These are called ${\langle{I,s}\rangle}$-ary formulas.
Formulas are constructed by using logical connectives within a fiber and logical flow between fibers.
\begin{itemize}
{\footnotesize{
\item[\textsf{fiber:}]
Let ${\langle{I,s}\rangle}$ be any type list.
Any ${\langle{I,s}\rangle}$-ary relation symbol is an (atomic) ${\langle{I,s}\rangle}$-ary formula;
that is, $R(I,s) \subseteq \widehat{R}(I,s)$.
For any pair of ${\langle{I,s}\rangle}$-ary formulas $\varphi$ and $\psi$, 
there are the following ${\langle{I,s}\rangle}$-ary formulas:
meet $(\varphi{\,\wedge\,}\psi)$, join $(\varphi{\,\vee\,}\psi)$,
implication $(\varphi{\,\rightarrowtriangle\,}\psi)$ and 
difference $(\varphi{\,\setminus\,}\psi)$. 
For any ${\langle{I,s}\rangle}$-ary formula $\varphi$,
there is an ${\langle{I,s}\rangle}$-ary negation formula $(\neg\varphi)$.
\item[\textsf{flow:}]
Let ${\langle{I',s'}\rangle} \xrightarrow{h} {\langle{I,s}\rangle}$ be any type list morphism.
For any ${\langle{I,s}\rangle}$-ary formula $\varphi$,
there are ${\langle{I',s'}\rangle}$-ary existentially/universally quantified formulas
${\scriptstyle\sum}_{t}(\varphi)$ and ${\scriptstyle\prod}_{t}(\varphi)$.
For any ${\langle{I',s'}\rangle}$-ary formula $\varphi'$,
there is a ${\langle{I,s}\rangle}$-ary substitution formula
${t}^{\ast}(\varphi') = \varphi'(t)$.
}}
\end{itemize}
%

%

%
\begin{table}
\begin{center}
\begin{tabular}[t]{r@{\hspace{10pt}}l}
\shortstack{\scriptsize{formula flow}\\\scriptsize{logical aspect}}
&
$\left\{\;\;\;\;\;\;\text{
{\scriptsize{$\begin{array}{r@{\hspace{10pt}}r@{\hspace{5pt}}c@{\hspace{5pt}}l@{\hspace{-5pt}}l}
\text{term vector} &
{\langle{I',s'}\rangle} & \xrightharpoondown{t} & {\langle{I,s}\rangle}
&
\text{ in }
\mathrmbf{Term}_{{\langle{X,\Omega}\rangle}}
\\
\text{operation} &
\mathcal{A}^{\ast}(I',s') & \xleftarrow{\mathcal{A}^{\ast}(t)} & \mathcal{A}^{\ast}(I,s)
&
\\
\text{inverse image} &
\mathrmbf{Rel}_{\mathcal{A}}(I',s') & \xrightarrow{\;\;{t}^{\ast}\;} & \mathrmbf{Rel}_{\mathcal{A}}(I,s)
&
\\
\text{quantification} &
\mathrmbf{Rel}_{\mathcal{A}}(I',s') 
& \xleftarrow[\;\;{\scriptstyle\forall}_{t}\;]{\;\;{\scriptstyle\exists}_{t}\;}
& \mathrmbf{Rel}_{\mathcal{A}}(I,s)
&
\\
\end{array}$}}}
\right.$
\\
& \\
\multicolumn{2}{c}{\mbox{}\hspace{110pt}{\Large{$\Uparrow\;\;$}}\shortstack{\scriptsize{functional}\\\scriptsize{aspect}}}
\\
& \\
\shortstack{\scriptsize{formula flow}\\\scriptsize{relational aspect}}
&
$\left\{\text{
{\scriptsize{$\begin{array}{r@{\hspace{10pt}}r@{\hspace{5pt}}c@{\hspace{5pt}}l@{\hspace{-5pt}}l}
\text{type list morphism} &
{\langle{I',s'}\rangle} & \xrightarrow{h} & {\langle{I,s}\rangle}
&
\text{ in }
\mathrmbf{List}(X)=\mathrmbf{Term}_{{\langle{X,\emptyset}\rangle}}
\\
\text{tuple map} &
\mathrmbfit{tup}_{\mathcal{E}}(I',s') & \xleftarrow{\mathrmbfit{tup}_{\mathcal{E}}(h)} & \mathrmbfit{tup}_{\mathcal{E}}(I,s)
&
\\
\text{inverse image} &
\mathrmbf{Rel}_{\mathcal{E}}(I',s') & \xrightarrow{\;\;{h}^{\ast}\;} & \mathrmbf{Rel}_{\mathcal{E}}(I,s)
&
\\
\text{quantification} &
\mathrmbf{Rel}_{\mathcal{E}}(I',s') 
& \xleftarrow[\;\;{\scriptstyle\forall}_{h}\;]{\;\;{\scriptstyle\exists}_{h}\;}
& \mathrmbf{Rel}_{\mathcal{E}}(I,s)
&
\\
\end{array}$}}}
\right.$
\\ & \\
\multicolumn{2}{c}{{\scriptsize{\begin{tabular}{p{280pt}}
When the relational aspect is lifted along the functional aspect to the first-order aspect
(Fig.~\ref{fbr:arch} of Section~\ref{sec:arch}),
formula flow is lifted 
from being along type list morphisms 
${\langle{I',s'}\rangle}\xrightarrow{h}{\langle{I,s}\rangle}$
to being along term vectors
${\langle{I',s'}\rangle}\xrightharpoondown{t}{\langle{I,s}\rangle}$.
This holds for 
formula definition (above),
formula function definition (Table~\ref{tbl:fmla:fn}),
formula axiomatization (Table~\ref{tbl:axioms}),
formula classification definition (Table~\ref{tbl:fmla:cls}),
satisfaction (Table~\ref{satisfaction}),
transformation to databases (Appendix~\ref{append:cls2db}),
etc.
\end{tabular}}}}
\end{tabular}
\end{center}
\caption{Lifting Flow}
\label{lift:flow}
\end{table}

\paragraph{Formula Fiber Passage.}

A schema morphism 
$\mathcal{S}_{2}\stackrel{{\langle{r,f}\rangle}}{\Longrightarrow}\mathcal{S}_{1}$
can be extended to a formula schema morphism 
$\mathrmbfit{fmla}(r,f) = {\langle{\hat{r},f}\rangle} :
\mathrmbfit{fmla}(\mathcal{S}_{2}) = {\langle{\widehat{R}_{2},\hat{\sigma}_{2},X_{2}}\rangle} \Longrightarrow
{\langle{\widehat{R}_{1},\hat{\sigma}_{1},X_{1}}\rangle} = \mathrmbfit{fmla}(\mathcal{S}_{1})$.
The formula function $\hat{r} : \widehat{R}_{2} \rightarrow \widehat{R}_{1}$,
which satisfies the condition
${inc}_{\mathcal{S}_{2}} \cdot \hat{r} = r \cdot {inc}_{\mathcal{S}_{1}}$,
is recursively defined in Table~\ref{tbl:fmla:fn}.
\begin{table}
\begin{center}
{\scriptsize{\setlength{\extrarowheight}{2pt}\begin{tabular}{|r@{\hspace{20pt}}l@{\hspace{10pt}$=$\hspace{10pt}}l|}
\multicolumn{3}{l}{\textsf{fiber:} type list ${\langle{I_{2},s_{2}}\rangle}$}
\\ \hline
\textit{operator} & \multicolumn{1}{l}{} & 
\\
relation
& $\hat{r}(r_{2})$
& $r(r_{2})$
\\
meet
& $\hat{r}(\varphi_{2}{\,\wedge_{{\langle{I_{2},s_{2}}\rangle}}\,}\psi_{2})$
& $(\hat{r}(\varphi_{2}){\,\wedge_{{\scriptscriptstyle\sum}_{f}(I_{2},s_{2})}\,}\hat{r}(\psi_{2}))$
\\
join
& $\hat{r}(\varphi_{2}{\,\vee_{{\langle{I_{2},s_{2}}\rangle}}\,}\psi_{2})$
& $(\hat{r}(\varphi_{2}){\,\vee_{{\scriptscriptstyle\sum}_{f}(I_{2},s_{2})}\,}\hat{r}(\psi_{2}))$
\\
negation
& $\hat{r}(\neg_{{\langle{I_{2},s_{2}}\rangle}}\,\varphi)$
& $\neg_{{\scriptscriptstyle\sum}_{f}(I_{2},s_{2})}\,\hat{r}(\varphi)$
\\
implication
& $\hat{r}(\varphi{\,\rightarrowtriangle_{{\langle{I_{2},s_{2}}\rangle}}\,}\psi)$
& $\hat{r}(\varphi){\,\rightarrowtriangle_{{\scriptscriptstyle\sum}_{f}(I_{2},s_{2})}\,}\hat{r}(\psi)$
\\
difference
& $\hat{r}(\varphi{\,\setminus_{{\langle{I_{2},s_{2}}\rangle}}\,}\psi)$
& $\hat{r}(\varphi){\,\setminus_{{\scriptscriptstyle\sum}_{f}(I_{2},s_{2})}\,}\hat{r}(\psi)$
\\ \hline
\multicolumn{3}{l}{}
\\
\multicolumn{3}{l}{\textsf{flow:} type list morphism 
${\langle{I_{2}',s_{2}'}\rangle} \xrightarrow{h} {\langle{I_{2},s_{2}}\rangle}$}
\\ \hline
\textit{operator} & \multicolumn{1}{l}{} & 
\\
existential
& $\hat{r}({\scriptstyle\sum}_{h}(\varphi_{2}))$
& ${\scriptstyle\sum}_{h}(\hat{r}(\varphi_{2}))$ 
\\
universal
& $\hat{r}({\scriptstyle\prod}_{h}(\varphi_{2}))$
& ${\scriptstyle\prod}_{h}(\hat{r}(\varphi_{2}))$ 
\\
substitution
& $\hat{r}({h}^{\ast}(\varphi_{2}'))$
& ${h}^{\ast}(\hat{r}(\varphi_{2}'))$ 
\\ \hline
\end{tabular}}}
\end{center}
\caption{Formula Function}
\label{tbl:fmla:fn}
\end{table}
%
%
%
\begin{proposition}
There is an idempotent formula passage
$\mathrmbfit{fmla} : \mathrmbf{Sch} \rightarrow \mathrmbf{Sch}$
that forms a monad ${\langle{\mathrmbf{Sch},\eta,\mathrmbfit{fmla}}\rangle}$ with embedding.
\end{proposition}




\paragraph{\underline{Relational Formalism}: $\mathrmbf{Fmla}$.}

\paragraph{Constraints.}


Let $\mathcal{S} = {\langle{R,\sigma,X}\rangle}$ be a relational schema.
A (binary) $\mathcal{S}$-sequent is a pair of formulas
$\varphi,\psi\in\widehat{R}$
with the same type list 
$\widehat{\sigma}(\varphi) = {\langle{I,s}\rangle} = \widehat{\sigma}(\psi)$. 
\footnote{We regard the formulas $\widehat{R}$ to be a set of types.
Since conjunction and disjunction are used in formulas, 
we can restrict attention to binary sequents.}
We represent a sequent using the turnstyle notation
$\varphi{\;\vdash\;}\psi$,
since we want a sequent to assert logical entailment.
A sequent expresses interpretation widening,
with the interpretation of $\varphi$ required to be within the interpretation of $\psi$.
We require entailment to be a preorder,
satisfying reflexivity and transitivity (Table~\ref{tbl:axioms}).
Hence,
for each type list ${\langle{I,s}\rangle}$
there is a fiber preorder
$\mathrmbf{Fmla}_{\mathcal{S}}(I,s) = {\langle{\widehat{R},\vdash}\rangle}$
consisting of all $\mathcal{S}$-formulas with this type list.
In first-order logic, 
we further require satisfaction of sufficient conditions 
(Table~\ref{tbl:axioms})
to described the various logical operations 
(connectives, quantifiers, etc.) 
used to build formulas.
%
%
%
An indexed $\mathcal{S}$-formula ${\langle{I,s,\varphi}\rangle}$
consists of a type list ${\langle{I,s}\rangle}$
and a formula $\varphi$ with signature ${\langle{I,s}\rangle}$.
%
An $\mathcal{S}$-constraint 
${\langle{I',s',\varphi'}\rangle}\xrightarrow{h}{\langle{I,s,\varphi}\rangle}$
consists of a type list morphism
${\langle{I',s'}\rangle} \xrightarrow{h} {\langle{I,s}\rangle}$
and 
a binary sequent $({\scriptstyle\sum}_{h}(\varphi){\;\vdash\;}\varphi')$,
or equivalently
a binary sequent  $(\varphi{\;\vdash\;}{h}^{\ast}(\varphi'))$.
The mathematical context $\mathrmbf{Fmla}(\mathcal{S})$ 
has indexed $\mathcal{S}$-formula as objects and $\mathcal{S}$-constraints as morphisms.
%
\footnote{In some sense,
this formula/constraint approach to formalism turns the tuple calculus upside down,
with atoms in the tuple calculus becoming constraints here.}
%
%
%
%
Let $\mathcal{S}_{2}\stackrel{{\langle{r,f}\rangle}}{\Longrightarrow}\mathcal{S}_{1}$
be a schema morphism.
We assume that the function map
$\widehat{R}_{2}\xrightarrow{\widehat{r}}\widehat{R}_{1}$
is monotonic
(Table~\ref{tbl:axioms}).
%
%
Hence,
there is 
a fibered formula passage
$\mathrmbf{Fmla}(\mathcal{S}_{2}) \xrightarrow{\mathrmbfit{fmla}_{{\langle{r,f}\rangle}}} \mathrmbf{Fmla}(\mathcal{S}_{1})$
that commutes with the type list projections
(Figure~\ref{fig:indexed:fibered}).
\comment{
\begin{itemize}
\item 
An $\mathcal{S}_{2}$-formula
${\langle{I_{2},s_{2},\varphi_{2}}\rangle}$ 
is mapped to the $\mathcal{S}_{1}$-formula
${\langle{{\scriptstyle\sum}_{f}(I_{2},s_{2}),\widehat{r}(\varphi_{2})}\rangle}$. 
%
\item 
An $\mathcal{S}_{2}$-constraint
${\langle{I_{2},s_{2},\varphi_{2}}\rangle}\xrightarrow{h}{\langle{I_{2}',s_{2}',\varphi_{2}'}\rangle}$ 
with binary sequent 
$({\scriptstyle\sum}_{h}(\varphi_{2}'){\;\vdash\;}\varphi_{2})$
is mapped to the $\mathcal{S}_{1}$-constraint
${\langle{{\scriptstyle\sum}_{f}(I_{2},s_{2}),\widehat{r}(\varphi_{2})}\rangle} \xrightarrow{h}
{\langle{{\scriptstyle\sum}_{f}(I_{2}',s_{2}'),\widehat{r}(\varphi_{2}')}\rangle}$
with binary sequent 
$({\scriptstyle\sum}_{h}(\widehat{r}(\varphi_{2}')){\;\vdash\;}\widehat{r}(\varphi_{2}))
=(\widehat{r}({\scriptstyle\sum}_{h}(\varphi_{2}')){\;\vdash\;}\widehat{r}(\varphi_{2}))$.
\end{itemize}
}
\begin{figure}
\begin{center}
\begin{tabular}{c@{\hspace{75pt}}c}
\begin{tabular}{c}
\setlength{\unitlength}{0.6pt}
\begin{picture}(120,110)(0,-10)
\put(0,80){\makebox(0,0){\footnotesize{$\mathrmbf{List}(X_{2})$}}}
\put(120,80){\makebox(0,0){\footnotesize{$\mathrmbf{List}(X_{1})$}}}
\put(60,0){\makebox(0,0){\footnotesize{$\mathrmbf{Pre}$}}}
\put(60,90){\makebox(0,0){\scriptsize{${\scriptstyle\sum}_{f}$}}}
\put(15,40){\makebox(0,0)[r]{\scriptsize{${\mathrmbfit{fmla}}_{\mathcal{S}_{2}}$}}}
\put(105,40){\makebox(0,0)[l]{\scriptsize{${\mathrmbfit{fmla}}_{\mathcal{S}_{2}}$}}}
\put(60,55){\makebox(0,0){\shortstack{\scriptsize{$\widehat{r}$}\\\large{$\Rightarrow$}}}}
\put(40,80){\vector(1,0){40}}
\put(10,65){\vector(3,-4){40}}
\put(110,65){\vector(-3,-4){40}}
\end{picture}
\end{tabular}
&
\begin{tabular}{c}
\setlength{\unitlength}{0.6pt}
\begin{picture}(120,110)(0,-10)
\put(0,80){\makebox(0,0){\footnotesize{$\mathrmbf{Fmla}(\mathcal{S}_{2})$}}}
\put(120,80){\makebox(0,0){\footnotesize{$\mathrmbf{Fmla}(\mathcal{S}_{1})$}}}
\put(0,0){\makebox(0,0){\footnotesize{$\mathrmbf{List}(X_{2})$}}}
\put(120,0){\makebox(0,0){\footnotesize{$\mathrmbf{List}(X_{1})$}}}
\put(60,95){\makebox(0,0){\scriptsize{${\mathrmbfit{fmla}}_{\langle{r,f}\rangle}$}}}
\put(60,-15){\makebox(0,0){\scriptsize{${\scriptstyle\sum}_{f}$}}}
\put(-10,40){\makebox(0,0)[r]{\scriptsize{$\mathrmbfit{list}_{\mathcal{S}_{2}}$}}}
\put(130,40){\makebox(0,0)[l]{\scriptsize{$\mathrmbfit{list}_{\mathcal{S}_{1}}$}}}
\put(43,80){\vector(1,0){34}}
\put(40,0){\vector(1,0){40}}
\put(0,65){\vector(0,-1){50}}
\put(120,65){\vector(0,-1){50}}
\end{picture}
\end{tabular}
\\ & \\
{\itshape indexed} & {\itshape fibered}
\end{tabular}
\end{center}
\caption{Indexed-Fibered}
\label{fig:indexed:fibered}
\end{figure}
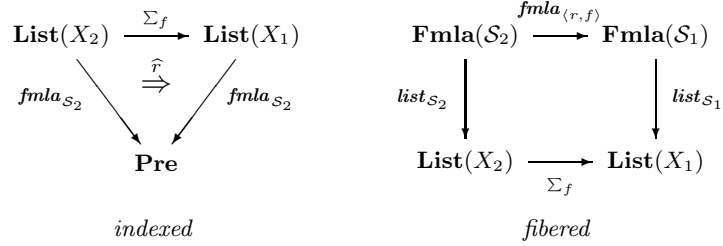
\begin{table}
\begin{center}
{\scriptsize{\setlength{\extrarowheight}{3.5pt}\begin{tabular}{|r@{\hspace{3pt}:\hspace{15pt}}l|}
\multicolumn{2}{l}{$\text{schema:}\;\mathcal{S}$} \\ 
\multicolumn{2}{l}{\textsf{fiber:} type list ${\langle{I,s}\rangle}$}
\\ \hline
reflexivity 
&
$\varphi{\;\vdash\;}\varphi$
\\
transitivity
&
$\varphi{\;\vdash\;}\varphi'$ and $\varphi'{\;\vdash\;}\varphi''$ implies $\varphi{\;\vdash\;}\varphi''$
\\ \hline\hline
meet
&
$\psi{\;\vdash\;}(\varphi{\;\wedge\;}\varphi')$ 
iff 
$\psi{\;\vdash\;}\varphi$ 
and
$\psi{\;\vdash\;}\varphi'$ 
\\
\multicolumn{1}{|c}{}
&
$(\varphi{\;\wedge\;}\varphi'){\;\vdash\;}\varphi$, 
$(\varphi{\;\wedge\;}\varphi'){\;\vdash\;}\varphi'$ 
\\ \hline
join
&
$(\varphi{\;\vee\;}\varphi'){\;\vdash\;}\psi$ 
iff 
$\varphi{\;\vdash\;}\psi$ 
and
$\varphi'{\;\vdash\;}\psi$ 
\\
\multicolumn{1}{|c}{}
&
$\varphi'{\;\vdash\;}(\varphi{\;\vee\;}\varphi)$, 
$\varphi'{\;\vdash\;}(\varphi{\;\vee\;}\varphi')$ 
\\ \hline
implication
&
$(\varphi{\;\wedge\;}\varphi'){\;\vdash\;}\psi$
iff
$\varphi{\;\vdash\;}(\varphi'{\rightarrowtriangle\,}\psi)$
\\ \hline
negation
&
$\neg\,(\neg\,(\varphi)){\;\vdash\;}\varphi$ 
\\ \hline
\multicolumn{2}{l}{\rule{0pt}{12pt}\textsf{flow:} type list morphism 
${\langle{I',s'}\rangle} \xrightarrow{h} {\langle{I,s}\rangle}$}
\\ \hline
${\scriptstyle\sum}_{h}$-monotonicity
&
$\varphi'{\;\vdash'\;}\psi'$ 
implies
${\scriptstyle\sum}_{h}(\varphi'){\;\vdash\;}{\scriptstyle\sum}_{h}(\psi')$
\\
${h}^{\ast}$-monotonicity
&
$\varphi{\;\vdash\;}\psi$ 
implies
${h}^{\ast}(\varphi){\;\vdash'\;}{h}^{\ast}(\psi)$
\\
${\scriptstyle\prod}_{h}$-monotonicity
&
$\varphi'{\;\vdash'\;}\psi'$ 
implies
${\scriptstyle\prod}_{h}(\varphi'){\;\vdash\;}{\scriptstyle\prod}_{h}(\psi')$
\\ \hline\hline
adjointness
&
${\scriptstyle\sum}_{h}(\varphi'){\;\vdash\;}\psi$
iff
$\varphi'{\;\vdash'\;}{h}^{\ast}(\psi)$
\\
\multicolumn{1}{|c}{}
&
$\varphi'{\;\vdash'\;}{h}^{\ast}({\scriptstyle\sum}_{h}(\varphi'))$,
${\scriptstyle\sum}_{h}({h}^{\ast}(\varphi)){\;\vdash\;}\varphi$
\\ \hline
\multicolumn{2}{l}{\rule{0pt}{16pt}$\text{schema morphism:}\;\mathcal{S}_{2}\stackrel{{\langle{r,f}\rangle}}{\Longrightarrow}\mathcal{S}_{1}$} 
\\ \hline 
$\widehat{r}$-monotonicity
&
$(\varphi_{2}{\;\vdash_{2}\;}\psi_{2})$ 
implies
$(\widehat{r}(\varphi_{2}){\;\vdash_{1}\;}\widehat{r}(\psi_{2}))$ 
\\ \hline
\end{tabular}}}
\end{center}
\caption{Axioms}
\label{tbl:axioms}
\end{table}
%

\paragraph{Specifications.}

A specification
$\mathcal{T} = {\langle{\mathcal{S},T}\rangle}$
consists of 
a schema $\mathcal{S} = {\langle{R,\sigma,X}\rangle}$ and 
a subset $T\subseteq\mathrmbf{Fmla}(\mathcal{S})$ of $\mathcal{S}$-constraints.
As a subgraph, 
$T$ extends to its consequence
$T^{\scriptscriptstyle\bullet}\subseteq\mathrmbf{Fmla}(\mathcal{S})$,
a mathematical subcontext,
by using paths of constraints.
A specification morphism
$\mathcal{T}_{2} = {\langle{\mathcal{S}_{2},T_{2}}\rangle}
\xrightarrow{\langle{r,f}\rangle}
{\langle{\mathcal{S}_{1},T_{1}}\rangle} = \mathcal{T}_{1}$
is a schema morphism  
$\mathcal{S}_{2}
\stackrel{{\langle{r,f}\rangle}}{\Longrightarrow}
\mathcal{S}_{1}$
that preserves constraints:
if sequent 
$\varphi_{2}'{\;\vdash\;}{h}^{\ast}(\varphi_{2})$
is asserted in $T_{2}$,
then sequent 
$\widehat{r}(\varphi_{2}'){\;\vdash\;}{h}^{\ast}(\widehat{r}(\varphi_{2}))$
is asserted in $T_{1}$.


\paragraph{\underline{First-order Linguistics}: 
$\underset{\mathrmbf{Sch}{\times}_{\mathrmbf{Set}}\mathrmbf{Oper}}{\mathrmbf{Lang}}
\hspace{-12pt}\xrightarrow{\mathrmbfit{sch}}\mathrmbf{Sch}$.}

%
A first-order logic (FOL) language 
$\mathcal{L} = {\langle{\mathcal{S},\mathcal{O}}\rangle}$
consists of a relational schema $\mathcal{S}={\langle{R,\sigma,X}\rangle}$ and 
an operator domain $\mathcal{O}={\langle{X,\Omega}\rangle}$
that share a common type set $X$.
%
%
A first-order logic (FOL) language morphism
$\mathcal{L}_{2} = {\langle{\mathcal{S}_{2},\mathcal{O}_{2}}\rangle}
\xrightarrow{\langle{r,f,\omega}\rangle}
{\langle{\mathcal{S}_{1},\mathcal{O}_{1}}\rangle} = \mathcal{L}_{1}$ 
consists of a relational schema morphism 
$\mathcal{S}_{2}
\xrightarrow{\langle{r,f}\rangle} 
\mathcal{S}_{1}$ 
and a functional language morphism
$\mathcal{O}_{2}
\xrightarrow{\langle{f,\omega}\rangle} 
\mathcal{O}_{1}$ 
that share a common type function 
$X_{2} \xrightarrow{f} X_{1}$.
%

\paragraph{\underline{First-order Formalism.}} 


%
A first-order specification
$\mathcal{T} = {\langle{\mathcal{S},T,\mathcal{O},E}\rangle}$
is an FOL language 
$\mathcal{L} = {\langle{\mathcal{S},\mathcal{O}}\rangle}$,
where
${\langle{\mathcal{S},T}\rangle}$
is a relational specification
and 
${\langle{\mathcal{O},E}\rangle}$
is an equational presentation. 
A first-order specification morphism
$\mathcal{T}_{2} = {\langle{\mathcal{S}_{2},T_{2},\mathcal{O}_{2},E_{2}}\rangle}
\xrightarrow{\langle{r,f}\rangle}
{\langle{\mathcal{S}_{1},T_{1},\mathcal{O}_{1},E_{1}}\rangle} = \mathcal{T}_{1}$
is an FOL language morphism
$\mathcal{L}_{2} = {\langle{\mathcal{S}_{2},\mathcal{O}_{2}}\rangle}
\xrightarrow{\langle{r,f,\omega}\rangle}
{\langle{\mathcal{S}_{1},\mathcal{O}_{1}}\rangle} = \mathcal{L}_{1}$, 
where
${\langle{\mathcal{S}_{2},T_{2}}\rangle}
\xrightarrow{\langle{r,f}\rangle}
{\langle{\mathcal{S}_{1},T_{1}}\rangle}$
is a relational specification morphism
and
${\langle{\mathcal{O}_{2},E_{2}}\rangle} \xrightarrow{{\langle{f,\omega}\rangle}} {\langle{\mathcal{O}_{1},E_{1}}\rangle}$
is a morphism of equational presentations.
A first-order specification morphism preserves constraints:
if sequent 
$\varphi_{2}'{\;\vdash\;}{[t]}^{\ast}(\varphi_{2})$
is asserted in $T_{2}$,
then sequent 
$\widehat{r}(\varphi_{2}'){\;\vdash\;}{[t]}^{\ast}(\widehat{r}(\varphi_{2}))$
is asserted in $T_{1}$.

\subsubsection{Semantics.}\label{sec:sem}

\paragraph{\underline{Relational Semantics}: 
$\mathrmbf{Rel}\xrightarrow{\mathrmbfit{sch}}\mathrmbf{Sch}$.}

\paragraph{Structures.}

A (model-theoretic) relational structure (classification form)
(IFF~\cite{iff})
$\mathcal{M} = {\langle{\mathcal{R},{\langle{\sigma,\tau}\rangle},\mathcal{E}}\rangle}$ 
is a hypergraph of classifications 
--- a two dimensional construction consisting of
a relation classification
$\mathcal{R} = {\langle{R,K,\models_{\mathcal{R}}}\rangle}$,
an entity classification 
$\mathcal{E} = {\langle{X,Y,\models_{\mathcal{E}}}\rangle}$ and
a list designation 
${\langle{\sigma,\tau}\rangle} : \mathcal{R} \rightrightarrows \mathrmbf{List}(\mathcal{E})$.
%
\footnote{$\mathrmbf{List}(\mathcal{E}) 
= {\langle{\mathrmbf{List}(X),\mathrmbf{List}(Y),\models_{\mathrmbf{List}(\mathcal{E})}}\rangle}$
is the list construction of the entity classification.
A tuple 
${\langle{J,t}\rangle} \in \mathrmbf{List}(Y)$
is classified by a signature
${\langle{I,s}\rangle} \in \mathrmbf{List}(X)$,
symbolized by ${\langle{J,t}\rangle} \models_{\mathrmbf{List}(\mathcal{E})} {\langle{I,s}\rangle}$,
when $J=I$ and $t_{i} \models_{\mathcal{E}} s_{i}$ for all $i \in I$.}
%
%
Hence,
a structure satisfies the following condition:
$k{\;\models_{\mathcal{R}}\;}r$
implies
$\tau(k){\;\models_{\mathrmbf{List}(\mathcal{E})}\;}\sigma(r)$.
%
%
A structure $\mathcal{M}$ has an associated schema 
$\mathrmbfit{sch}(\mathcal{M}) = {\langle{R,\sigma,X}\rangle}$.

\paragraph{Formulas.}

Any structure 
$\mathcal{M} = {\langle{\mathcal{R},{\langle{\sigma,\tau}\rangle},\mathcal{E}}\rangle}$ 
has an associated formula structure
$\mathrmbfit{fmla}(\mathcal{M}) 
= {\langle{\widehat{\mathcal{R}},{\langle{\widehat{\sigma},\tau}\rangle},\mathcal{E}}\rangle}$
with schema 
$\mathrmbfit{sch}(\mathrmbfit{fmla}(\mathcal{M})) 
= {\langle{\widehat{\mathcal{R}},\widehat{\sigma},X}\rangle}$.
The formula classification
$\widehat{\mathcal{R}} = {\langle{\widehat{R},K,\models_{\widehat{\mathcal{R}}}}\rangle}$,
which extends the relation classification of $\mathcal{M}$,
is directly defined by induction
in Table~\ref{tbl:fmla:cls}.
\begin{table}
\begin{center}
{\scriptsize{\setlength{\extrarowheight}{2pt}\begin{tabular}{|r@{\hspace{20pt}}l@{\hspace{10pt}\underline{when}\hspace{10pt}}c|}
\multicolumn{3}{l}{\textsf{fiber:} type list ${\langle{I,s}\rangle}$ with interpretation 
$\mathrmbfit{tup}_{\mathcal{E}}(I,s){\,=\,}\prod_{i\in{I}}\,\mathrmbfit{ext}_{\mathcal{E}}(s_{i})$}
\\ \hline
\textit{operator} & \multicolumn{1}{l}{\textit{definiendum}} & \multicolumn{1}{c|}{\textit{definiens}}
\\
relation
& $k{\;\models_{\widehat{\mathcal{R}}}\;}r$
& $k{\;\models_{\mathcal{R}}\;}r$
\\
meet
& $k{\;\models_{\widehat{\mathcal{R}}}\;}(\varphi{\,\wedge\,}\psi)$
& $k{\;\models_{\widehat{\mathcal{R}}}\;}\varphi$ and $k{\;\models_{\widehat{\mathcal{R}}}\;}\psi$
\\
join
& $k{\;\models_{\widehat{\mathcal{R}}}\;}(\varphi{\,\vee\,}\psi)$
& $k{\;\models_{\widehat{\mathcal{R}}}\;}\varphi$ or $k{\;\models_{\widehat{\mathcal{R}}}\;}\psi$
\\
top
& \multicolumn{2}{l|}{$k{\;\models_{\widehat{\mathcal{R}}}\;}{\scriptstyle\top}$}
\\
bottom
& \multicolumn{2}{l|}{$k{\;\cancel{\models}_{\widehat{\mathcal{R}}}\;}{\scriptstyle\bot}$}
\\
negation
& $k{\;\models_{\widehat{\mathcal{R}}}\;}(\neg\varphi)$
& $k{\;\cancel{\models}_{\widehat{\mathcal{R}}}\;}\varphi$
\\
implication
& $k{\;\models_{\widehat{\mathcal{R}}}\;}(\varphi{\,\rightarrowtriangle\,}\psi)$
& if $k{\;\models_{\widehat{\mathcal{R}}}\;}\varphi$ then $k{\;\models_{\widehat{\mathcal{R}}}\;}\psi$
\\
difference
& $k{\;\models_{\widehat{\mathcal{R}}}\;}(\varphi{\,\setminus\,}\psi)$
& $k{\;\models_{\widehat{\mathcal{R}}}\;}\varphi$ but not $k{\;\models_{\widehat{\mathcal{R}}}\;}\psi$
\\ \hline
\multicolumn{3}{l}{\rule{0pt}{28pt}\textsf{flow:} type list morphism
$\overset{\textstyle\widehat{\sigma}(\varphi')}{\overbrace{\langle{I',s'}\rangle}}
\xrightarrow{h}
\overset{\textstyle\widehat{\sigma}(\varphi)}{\overbrace{\langle{I,s}\rangle}}$ with interpretation
$\mathrmbfit{tup}_{\mathcal{E}}(I',s')\xleftarrow{\mathrmbfit{tup}_{\mathcal{E}}(h)}\mathrmbfit{tup}_{\mathcal{E}}(I,s)$}
\\ \hline
\textit{operator} & \multicolumn{1}{l}{\textit{definiendum}} & \multicolumn{1}{c|}{\textit{definiens}}
\\
existential
& $k{\;\models_{\widehat{\mathcal{R}}}\;}{\scriptstyle\sum}_{h}(\varphi)$
&
$\tau(k){\,\in\,}{\exists}_{h}(\mathrmbfit{R}_{\widehat{\mathcal{M}}}(\varphi))$ 
\\
universal
& $k{\;\models_{\widehat{\mathcal{R}}}\;}{\scriptstyle\prod}_{h}(\varphi)$
& 
$\tau(k){\,\in\,}{\forall}_{h}(\mathrmbfit{R}_{\widehat{\mathcal{M}}}(\varphi))$ 
\\
substitution
& $k{\;\models_{\widehat{\mathcal{R}}}\;}{h}^{\ast}(\varphi')$
& 
$\tau(k){\,\in\,}{h}^{-1}(\mathrmbfit{R}_{\widehat{\mathcal{M}}}(\varphi'))$ 
\\
\multicolumn{2}{|r}{}
&
\multicolumn{1}{c|}{where $\mathrmbfit{R}_{\widehat{\mathcal{M}}}(\varphi)
={\wp}\tau(\mathrmbfit{ext}_{\widehat{\mathcal{R}}}(\varphi))$\rule[-5pt]{0pt}{10pt}}
\\ \hline
\end{tabular}}}
\end{center}
\caption{Formula Classification}
\label{tbl:fmla:cls}
\end{table}
%

\paragraph{Satisfaction.}
Satisfaction is defined in terms of the extent order of the formula classification.
For any $\mathcal{S}$-structure $\mathcal{M} \in \mathrmbf{Rel}(\mathcal{S})$,
two formula $\varphi,\psi\in\widehat{R}$ 
with the same type list $\sigma(\varphi)=\sigma(\psi)$
satisfy the specialization-generalization order
$\varphi{\;\leq_{\widehat{\mathcal{R}}}\;}\psi$
when
their extents satisfy the containment order
$\mathrmbfit{ext}_{\widehat{\mathcal{R}}}(\varphi){\;\subseteq\;}\mathrmbfit{ext}_{\widehat{\mathcal{R}}}(\psi)$.
%
An $\mathcal{S}$-structure $\mathcal{M} \in \mathrmbf{Rel}(\mathcal{S})$
satisfies 
an $\mathcal{S}$-sequent $(\varphi{\;\vdash\;}\psi)$
when
$\varphi{\;\leq_{\widehat{\mathcal{R}}}\;}\psi$.
%
An $\mathcal{S}$-structure $\mathcal{M} \in \mathrmbf{Rel}(\mathcal{S})$
satisfies 
an $\mathcal{S}$-constraint $\varphi'{\;\xrightarrow{h}\;}\varphi$,
symbolized by
$\mathcal{M}{\;\models_{\mathcal{S}}\;}(\varphi'{\;\xrightarrow{h}\;}\varphi)$,
when
$\mathcal{M}$ 
satisfies the sequent $({\scriptstyle\sum}_{h}(\varphi){\;\vdash\;}\varphi')$;
that is,
when
${\scriptstyle\sum}_{h}(\varphi){\;\leq_{\widehat{\mathcal{R}}}\;}\varphi'$;
equivalently,
when
$\varphi{\;\leq_{\widehat{\mathcal{R}}}\;}{h}^{\ast}(\varphi')$.
This can be expressed in terms of implication as
$({\scriptstyle\sum}_{h}(\varphi){\,\rightarrowtriangle\,}\varphi') \equiv \top$;
equivalently,
$(\varphi{\,\rightarrowtriangle\,}{h}^{\ast}(\varphi')) \equiv \top$.
When converting structures to databases,
the satisfaction relationship 
$\mathcal{M}{\;\models_{\mathcal{S}}\;}(\varphi\xrightarrow{h}\varphi')$
determines
the morphism of $\mathcal{E}$-relations
$\mathrmbfit{R}_{\widehat{\mathcal{M}}}(\varphi)\xleftarrow{h}\mathrmbfit{R}_{\widehat{\mathcal{M}}}(\varphi')$
in $\mathrmbf{Rel}(\mathcal{E})$
and
a morphism of $\mathcal{E}$-tables 
$\mathrmbfit{T}_{\widehat{\mathcal{M}}}(\varphi)\xleftarrow{{\langle{h,k}\rangle}}
\mathrmbfit{T}_{\widehat{\mathcal{M}}}(\varphi')$
in $\mathrmbf{Tbl}(\mathcal{E})$.
(The operators 
$\mathrmbfit{R}_{\widehat{\mathcal{M}}}$ and $\mathrmbfit{T}_{\widehat{\mathcal{M}}}$
are defined in Appendix~\ref{rel:interp}. Satisfaction is summarized in Table~\ref{satisfaction}.)
%
\begin{table}
\begin{center}
\begin{minipage}{320pt}
\begin{center}
{\fbox{\footnotesize{\begin{tabular}{r@{\hspace{4pt}}l}
       & 
$\mathcal{M}{\;\models_{\mathcal{S}}\;}(\varphi'\xrightarrow{h}\varphi)$
\\ when & 
${\scriptstyle\sum}_{h}(\varphi){\;\leq_{\widehat{\mathcal{R}}}\;}\varphi'$
\\ iff & 
$\forall_{k\,\in\,K}\left(\,k{\;\models_{\widehat{\mathcal{R}}}\;}({\scriptstyle\sum}_{h}(\varphi){\,\rightarrowtriangle\,}\varphi')\,\right)$
\\ iff & 
$\forall_{k\,\in\,K}\left(\,k{\;\models_{\widehat{\mathcal{R}}}\;}{\scriptstyle\sum}_{h}(\varphi)\;\text{implies}\;k{\;\models_{\widehat{\mathcal{R}}}\;}\varphi'\,\right)$
\\ implies & 
${\exists}_{h}(\mathrmbfit{R}_{\widehat{\mathcal{M}}}(\varphi)){\,\leq\,}\mathrmbfit{R}_{\widehat{\mathcal{M}}}(\varphi')$
\footnote{For relational structure
$\mathcal{M} = {\langle{\mathcal{R},{\langle{\sigma,\tau}\rangle},\mathcal{E}}\rangle}$,
the fibered mathematical context $\mathrmbf{Rel}(\mathcal{E})^{\mathrm{op}}\xrightarrow{\mathrmbfit{list}}\mathrmbf{List}(X)$ 
of $\mathcal{E}$-relations
is determined by 
the indexed preorder
$\mathrmbf{List}(X)^{\mathrm{op}}\xrightarrow{\mathrmbfit{rel}}\mathrmbf{Pre}$,
which maps a type list ${\langle{I,s}\rangle}$ to the fiber relational order
$\mathrmbf{Rel}_{\mathcal{E}}(I,s)
={\langle{{\wp}\mathrmbfit{tup}_{\mathcal{E}}(I,s),\subseteq}\rangle}$
and maps a type list morphism
${\langle{I',s'}\rangle} \xrightharpoondown{h} {\langle{I,s}\rangle}$
to the fiber monotonic function
${\exists}_{h} = {\exists}_{\mathrmbfit{tup}_{\mathcal{E}}(h)} :
\mathrmbf{Rel}_{\mathcal{E}}(I',s')\leftarrow\mathrmbf{Rel}_{\mathcal{E}}(I,s)$.
Similarly, for
the fibered context 
$\mathrmbf{Tbl}(\mathcal{E})^{\mathrm{op}}\xrightarrow{\mathrmbfit{pr}}\mathrmbf{Term}(X)$
of $\mathcal{E}$-tables.}
\\ implies & 
$\exists_{k}
\left(
{\scriptstyle\sum}_{h}(\mathrmbfit{T}_{\widehat{\mathcal{M}}}(\varphi)){\,\xrightarrow{k}\,}\mathrmbfit{T}_{\widehat{\mathcal{M}}}(\varphi')
\right)$
\end{tabular}}}}
\end{center}
\end{minipage}
\end{center}
\caption{Satisfaction}
\label{satisfaction}
\end{table}
%

\paragraph{Structure Morphisms.}

A (model-theoretic) structure morphism 
(IFF~\cite{iff})
\[\mbox{\footnotesize{$
{\langle{r,k,f,g}\rangle} : 
\mathcal{M}_{2} = {\langle{\mathcal{R}_{2},{\langle{\sigma_{2},\tau_{2}}\rangle},\mathcal{E}_{2}}\rangle} \rightleftarrows
{\langle{\mathcal{R}_{1},{\langle{\sigma_{1},\tau_{1}}\rangle},\mathcal{E}_{1}}\rangle} = \mathcal{M}_{1}
$}\normalsize}\]
is a two dimensional construction consisting of 
a relation infomorphism 
${\langle{r,k}\rangle} : 
\mathcal{R}_{2} = {\langle{R_{2},K_{2},\models_{\mathcal{R}_{2}}}\rangle} \rightleftarrows 
{\langle{R_{1},K_{1},\models_{\mathcal{R}_{1}}}\rangle} = \mathcal{R}_{1}$,
an entity infomorphism 
${\langle{f,g}\rangle} : 
\mathcal{E}_{2} = {\langle{X_{2},Y_{2},\models_{\mathcal{E}_{2}}}\rangle} \rightleftarrows 
{\langle{X_{1},Y_{1},\models_{\mathcal{E}_{1}}}\rangle} = \mathcal{E}_{1}$,
and a list classification square 
\vspace{-5pt}
\[\mbox{\footnotesize{
${\langle{{\langle{r,k}\rangle},\mathrmbf{List}_{{\langle{f,g}\rangle}}}\rangle} :
{\langle{\mathcal{R}_{2}
\!\!\!\stackrel{{\langle{\sigma_{2},\tau_{2}}\rangle}}{\rightrightarrows}\!\!\!
\mathrmbf{List}(\mathcal{E}_{2})}\rangle}
\rightleftarrows
{\langle{\mathcal{R}_{1} 
\!\!\!\stackrel{{\langle{\sigma_{1},\tau_{1}}\rangle}}{\rightrightarrows}\!\!\!
\mathrmbf{List}(\mathcal{E}_{1})}\rangle}$,
}\normalsize}\]
where the list infomorphism of the entity infomorphism
is the vertical target of the list square.
Hence,
a structure morphism satisfies the following conditions.
\begin{center}
{\footnotesize{\begin{tabular}{r@{\hspace{10pt}}c@{\hspace{10pt}}l}
\comment{
\multicolumn{3}{l}{\textsf{designations}}
\\
$k_{2}{\;\models_{\mathcal{R}_{2}}\;}r_{2}$
&
implies
&
$\tau_{2}(k_{2}){\;\models_{\mathrmbf{List}(\mathcal{E}_{2})}\;}\sigma_{2}(r_{2})$
\\
$k_{1}{\;\models_{\mathcal{R}_{1}}\;}r_{1}$
&
\underline{implies}
&
$\tau_{1}(k_{1}){\;\models_{\mathrmbf{List}(\mathcal{E}_{1})}\;}\sigma_{1}(r_{1})$
\\
\multicolumn{3}{l}{}
\\
}
\multicolumn{3}{l}{\textsf{infomorphisms}}
\\
$k_{1}{\;\models_{\mathcal{R}_{1}}\;}r(r_{2})$
&
\underline{iff}
&
$k(k_{1}){\;\models_{\mathcal{R}_{2}}\;}r_{2}$
\\
$y_{1}{\;\models_{\mathcal{E}_{1}}\;}f(x_{2})$
&
\underline{iff}
&
$g(y_{1}){\;\models_{\mathcal{E}_{2}}\;}x_{2}$
\\
${t_{1}{\,\cdot\,}g}={\scriptstyle\sum}_{g}(J,t_{1})
{\;\models_{\mathrmbf{List}(\mathcal{E}_{2})}\;}
{\langle{I,s_{2}}\rangle}=s_{2}$
&
\underline{iff}
&
$t_{1}={\langle{J,t_{1}}\rangle}
{\;\models_{\mathrmbf{List}(\mathcal{E}_{1})}\;}
{\scriptstyle\sum}_{f}(I,s_{2})={s_{2}{\,\cdot\,}f}$
\\
\multicolumn{3}{l}{}
\\
\multicolumn{3}{l}{\textsf{list preservation}}
\\
\multicolumn{3}{c}{$r{\;\cdot\;}\sigma_{1}\;=\;\sigma_{2}{\;\cdot\;}{\scriptstyle\sum}_{f}$}
\\
\multicolumn{3}{c}{$k{\;\cdot\;}\tau_{2}\;=\;\tau_{1}{\;\cdot\;}{\scriptstyle\sum}_{g}$}
\end{tabular}}}
\end{center}
Structure morphisms compose component-wise.
Let $\mathrmbf{Rel}$ denote the mathematical context of relational structures and structure morphisms.
%
A structure morphism
${\langle{r,k,f,g}\rangle} : \mathcal{M}_{2}\rightleftarrows\mathcal{M}_{1}$
has an associated schema morphism
$\mathrmbfit{sch}(r,k,f,g) = {\langle{r,f}\rangle} :
\mathrmbfit{sch}(\mathcal{M}_{2}) = {\langle{R_{2},\sigma_{2},X_{2}}\rangle} \Longrightarrow
{\langle{R_{1},\sigma_{1},X_{1}}\rangle} = \mathrmbfit{sch}(\mathcal{M}_{1})$.
Hence,
there is a schema passage
$\mathrmbfit{sch} : \mathrmbf{Rel} \rightarrow \mathrmbf{Sch}$.
\paragraph{Formula.}

Any structure morphism
${\langle{r,k,f,g}\rangle} : 
{\langle{\mathcal{R}_{2},{\langle{\sigma_{2},\tau_{2}}\rangle},\mathcal{E}_{2}}\rangle} \rightleftarrows
{\langle{\mathcal{R}_{1},{\langle{\sigma_{1},\tau_{1}}\rangle},\mathcal{E}_{1}}\rangle}$
has an associated formula structure morphism
\[\mbox{\footnotesize{$
\mathrmbfit{fmla}(r,k,f,g) = {\langle{\widehat{r},k,f,g}\rangle} : 
\mathrmbfit{fmla}(\mathcal{M}_{2}) = {\langle{\widehat{\mathcal{R}}_{2},{\langle{\sigma_{2},\tau_{2}}\rangle},\mathcal{E}_{2}}\rangle} \rightleftarrows
{\langle{\widehat{\mathcal{R}}_{1},{\langle{\sigma_{1},\tau_{1}}\rangle},\mathcal{E}_{1}}\rangle} = \mathrmbfit{fmla}(\mathcal{M}_{1})
$}\normalsize}\]
with schema morphism 
$\mathrmbfit{sch}(\mathrmbfit{fmla}(r,k,f,g)) = {\langle{\widehat{r},f}\rangle} :
{\langle{\widehat{R}_{2},\widehat{\sigma}_{2},X_{2}}\rangle} \Rightarrow
{\langle{\widehat{R}_{1},\widehat{\sigma}_{1},X_{1}}\rangle}
$.
Hence,
there is a formula passage
$\mathrmbfit{fmla} : \mathrmbf{Rel} \rightarrow \mathrmbf{Rel}$.
\footnote{The schema and formula passages commute:
$\mathrmbfit{fmla}{\;\circ\;}\mathrmbfit{sch} 
= \mathrmbfit{sch}{\;\circ\;}\mathrmbfit{fmla}$
(Fig.~\ref{fbr:arch}).}
Between any structure and its formula extension is an embedding structure morphism
$\eta_{\mathcal{M}} = {\langle{{inc}_{\mathcal{M}},{1}_{K},{1}_{\mathcal{E}}}\rangle} :
\mathcal{M} \Longrightarrow \mathrmbfit{fmla}(\mathcal{M})$.
The formula operator commutes with embedding:
$\eta_{\mathcal{M}_{2}}{\,\circ\,}\mathrmbfit{fmla}(r,k,f,g) = 
{\langle{r,k,f,g}\rangle}{\,\circ\,}\eta_{\mathcal{M}_{1}}$.
There is an embedding bridge
$\eta : \mathrmbfit{id}_{\mathrmbf{Rel}} \Rightarrow \mathrmbfit{fmla}$.
\begin{proposition}
There is an idempotent formula passage
$\mathrmbfit{fmla} : \mathrmbf{Rel} \rightarrow \mathrmbf{Rel}$
that forms a monad ${\langle{\mathrmbf{Rel},\eta,\mathrmbfit{fmla}}\rangle}$ with embedding.
\end{proposition}
\comment{
\begin{fact}
For any structure  morphism
${\langle{r,k,f,g}\rangle} : 
\mathcal{M}_{2} 
\rightleftarrows
\mathcal{M}_{1}$,
the formula function is monotonic between formula extent orders 
$\widehat{r} :
\mathrmbfit{ext}(\widehat{\mathcal{R}}_{2})
\rightarrow
\mathrmbfit{ext}(\widehat{\mathcal{R}}_{2}) 
$:
$(\varphi{\;\leq_{\widehat{\mathcal{R}_{2}}}\;}\psi)$
\underline{implies}
$(\widehat{r} (\varphi){\;\leq_{\widehat{\mathcal{R}_{1}}}\;}\widehat{r} (\psi))$.
\end{fact}
}

\paragraph{Structure Fiber Passage.}

Let
$\mathcal{S}_{2}={\langle{R_{2},\sigma_{2},X_{2}}\rangle} 
\xRightarrow{\langle{r,f}\rangle}
{\langle{R_{1},\sigma_{1},X_{1}}\rangle}=\mathcal{S}_{1}$
be a schema morphism.
There is a structure passage
$\mathrmbf{Rel}(\mathcal{S}_{2}) \xleftarrow{\mathrmbfit{rel}_{{\langle{r,f}\rangle}}} \mathrmbf{Rel}(\mathcal{S}_{2})$
defined as follows.
Let
$\mathcal{M}_{1} = {\langle{\mathcal{R}_{1},{\langle{\sigma_{1},\tau_{1}}\rangle},\mathcal{E}_{1}}\rangle}
\in \mathrmbf{Rel}(\mathcal{S}_{1})$
be an $\mathcal{S}_{1}$-structure
with
a relation classification
$\mathcal{R}_{1} = {\langle{R_{1},K_{1},\models_{\mathcal{R}_{1}}}\rangle}$,
an entity classification 
$\mathcal{E}_{1} = {\langle{X_{1},Y_{1},\models_{\mathcal{E}_{1}}}\rangle}$ and
a list designation 
${\langle{\sigma_{1},\tau_{1}}\rangle} : \mathcal{R}_{1} \rightrightarrows \mathrmbf{List}(\mathcal{E}_{1})$.
Define the inverse image
$\mathcal{S}_{2}$-structure
$\mathrmbfit{rel}_{{\langle{r,f}\rangle}}(\mathcal{M}_{1}) = 
{\langle{r^{-1}(\mathcal{R}_{1}),{\langle{\sigma_{2},\tau_{1}}\rangle},f^{-1}(\mathcal{E}_{1})}\rangle}
\in \mathrmbf{Rel}(\mathcal{S}_{2})$
with
$r^{-1}(\mathcal{R}_{1}) = {\langle{R_{2},K_{1},\models_{r}}\rangle}$,
$f^{-1}(\mathcal{E}_{1}) = {\langle{X_{2},Y_{1},\models_{f}}\rangle}$ and
a list designation 
${\langle{\sigma_{2},\tau_{1}}\rangle} : r^{-1}(\mathcal{R}_{1}) \rightrightarrows f^{-1}(\mathcal{E}_{1})$.
From the definitions of inverse image classifications, 
we have the two logical equivalences
(1) 
$k_{1} \models_{r} r_{2}
\;\text{\underline{iff}}\;
k_{1} \models_{\mathcal{E}_{1}} r(r_{2})$ 
and
(2)
${\langle{J_{1},t_{1}}\rangle} \models_{{\scriptscriptstyle\sum}_{f}} {\langle{I_{2},s_{2}}\rangle} 
\;\text{\underline{iff}}\;
{\langle{J_{1},t_{1}}\rangle} \models_{\mathrmbf{List}(\mathcal{E}_{1})} {\scriptstyle\sum}_{f}(I_{2},s_{2})$.
Hence,
$k_{1} \models_{r} r_{2}
\;\text{\underline{implies}}\; 
\tau_{1}(k_{1}) \models_{{\scriptscriptstyle\sum}_{f}} \sigma_{2}(r_{2})$.
There is a bridging structure morphism
\[\mbox
{\footnotesize{
$\mathrmbfit{rel}_{{\langle{r,f}\rangle}}(\mathcal{M}_{1}) 
= {\langle{r^{-1}(\mathcal{R}_{1}),{\langle{\sigma_{2},\tau_{1}}\rangle},f^{-1}(\mathcal{E}_{1})}\rangle}
\stackrel{{\langle{r,1_{K},f,1_{Y}}\rangle}}{\rightleftarrows} 
{\langle{\mathcal{R}_{1},{\langle{\sigma_{1},\tau_{1}}\rangle},\mathcal{E}_{1}}\rangle} = \mathcal{M}_{1}$
}\normalsize}
\]
with relation and entity infomorphisms
$r^{-1}(\mathcal{R}_{1})
\stackrel{{\langle{r,1_{K}}\rangle}}{\rightleftarrows}
\mathcal{R}_{1}$ 
and
$f^{-1}(\mathcal{E}_{1})
\stackrel{{\langle{f,1_{Y}}\rangle}}{\rightleftarrows} 
\mathcal{E}_{1}$.
%

\paragraph{\underline{First-order Semantics}: 
$\mathrmbf{Rel}\xleftarrow{\mathrmbfit{rel}}\hspace{-6pt}
\underset{\mathrmbf{Rel}{\times}_{\mathrmbf{Cls}}\mathrmbf{Alg}}{\mathrmbf{Struc}}
\hspace{-8pt}\xrightarrow{\mathrmbfit{lang}}\mathrmbf{Lang}$.}

\begin{sloppypar}
The mathematical context of first-order structures $\mathrmbf{Struc}$ is the product of 
the context $\mathrmbf{Rel}$ of relational structures and
the context $\mathrmbf{Alg}$ of algebras modulo
the context $\mathrmbf{Cls}$ of classifications.
%
%
A first-order logic (FOL) structure is a ``pair'' 
$\mathcal{M}={\langle{\mathcal{R},{\langle{\sigma,\tau}\rangle},\mathcal{E},{\langle{\Omega,A,\delta}\rangle}}\rangle}$
consisting of
a relational structure
${\langle{\mathcal{R},{\langle{\sigma,\tau}\rangle},\mathcal{E}}\rangle}$
and 
an algebra
${\langle{\mathcal{E},{\langle{\Omega,A,\delta}\rangle}}\rangle}$
that share a common entity classification $\mathcal{E}$.
The algebra is the semantic base and
the relational structure is the superstructure.
%
%
Given a FOL language $\mathcal{L} = {\langle{\mathcal{S},\mathcal{O}}\rangle}$
and an $\mathcal{L}$-structure $\mathcal{M}$
with
relational $\mathcal{S}$-structure
$\mathrmbfit{rel}(\mathcal{M})$
and
$\mathcal{O}$-algebra
$\mathrmbfit{alg}(\mathcal{M})$.
$\mathcal{M}$ satisfies 
an $\mathcal{L}$-equation
${\langle{I',s'}\rangle} \xrightharpoondown{(t=t')} {\langle{I,s}\rangle}$,
symbolized by
$\mathcal{M}{\;\models_{\mathcal{L}}\;}(t=t')$,
when
$\mathrmbfit{alg}(\mathcal{M}){\;\models_{\mathcal{L}}\;}(t=t')$; and
$\mathcal{M}$ satisfies 
an $\mathcal{L}$-constraint $\varphi'\xrightarrow{[t]}\varphi$,
symbolized by
$\mathcal{M}{\;\models_{\mathcal{L}}\;}(\varphi'\xrightarrow{[t]}\varphi)$,
when
$\mathrmbfit{rel}(\mathcal{M}){\;\models_{\mathcal{S}}\;}(\varphi'\xrightarrow{t}\varphi)$
for any representative term vector
$\widehat{\sigma}(\varphi')={\langle{I',s'}\rangle} \xrightharpoondown{t} {\langle{I,s}\rangle}=\widehat{\sigma}(\varphi)$.
%
%
A first-order logic (FOL) structure morphism
${\langle{\mathcal{R}_{2},{\langle{\sigma_{2},\tau_{2}}\rangle},\mathcal{E}_{2},{\langle{\Omega_{2},A_{2},\delta_{2}}\rangle}}\rangle}
	\xrightarrow{\langle{{\langle{r,k}\rangle},{\langle{f,g}\rangle},{\langle{\omega,h}\rangle}}\rangle}
{\langle{\mathcal{R}_{1},{\langle{\sigma_{1},\tau_{1}}\rangle},\mathcal{E}_{1},{\langle{\Omega_{1},A_{1},\delta_{1}}\rangle}}\rangle}$
consists
a relational structure morphism
${\langle{\mathcal{R}_{2},{\langle{\sigma_{2},\tau_{2}}\rangle},\mathcal{E}_{2}}\rangle}
	\xrightarrow{\langle{{\langle{r,k}\rangle},{\langle{f,g}\rangle}}\rangle}
{\langle{\mathcal{R}_{1},{\langle{\sigma_{1},\tau_{1}}\rangle},\mathcal{E}_{1}}\rangle}$
and 
an many-sorted algebraic homomorphism
${\langle{\mathcal{E}_{2},\mathcal{O}_{2},{\langle{A_{2},\delta_{2}}\rangle}}\rangle}
	\xrightarrow{\langle{f,g,\omega,h}\rangle}
{\langle{\mathcal{E}_{1},\mathcal{O}_{1},{\langle{A_{1},\delta_{1}}\rangle}}\rangle}$
that share a common entity infomorphism
${\langle{f,g}\rangle} : \mathcal{E}_{2} \rightleftarrows \mathcal{E}_{1}$.
\end{sloppypar}

\subsection{Examples}\label{sec:egs}

\begin{description}
%
\item[Conceptual Graphs:] 
Consider the English sentence
``John is going to Boston by bus''
\cite{sowa:kr}.
We describe its representation
in a {\ttfamily FOLE} logic language $\mathcal{L} = {\langle{R,\sigma,X,\Omega}\rangle}$.
By representing the verb as a ternary relation,
a graphical representation is
\begin{center}
\begin{tabular}{c}
\setlength{\unitlength}{0.45pt}
\begin{picture}(120,80)(0,-5)
\put(0,60){\makebox(0,0){\footnotesize{$[Person: John]\xleftarrow{agnt}(Go)\xrightarrow{dest}[City: Boston]$}}}
\put(25,30){\makebox(0,0){\small{$\downarrow$}\;\footnotesize{$inst$}}}
\put(10,0){\makebox(0,0){\footnotesize{$[Bus]$}}}
\end{picture}
\end{tabular}
\end{center}
Formally,
we have the following elements:
three entity types $\mathtt{Person},\mathtt{City},\mathtt{Bus}{\,\in\,}X$;
a relation type $\mathtt{Go}{\,\in\,}R$
with signature $\sigma(\mathtt{Go}) = {\langle{I,s}\rangle}$
having valence 3, arity $I = \{\mathtt{agnt},\mathtt{dest},\mathtt{inst}\}$
and signature function $I \xrightarrow{s} X$
mapping
$\mathtt{agnt}\mapsto\mathtt{Person}$,
$\mathtt{dest}\mapsto\mathtt{City}$,
$\mathtt{inst}\mapsto\mathtt{Bus}$;
a constant symbol $\mathtt{John}{\,\in\,}{\Omega}_{\mathtt{Person},{\langle{\emptyset,0_{X}}\rangle}}$ 
of sort $\mathtt{Person}$ 
and
a constant symbol $\mathtt{Boston}{\,\in\,}{\Omega}_{\mathtt{City},{\langle{\emptyset,0_{X}}\rangle}}$ 
of sort $\mathtt{City}$.
%
\footnote{\label{roles}
According to (Sowa~\cite{sowa:kr}),
every participant of a process is an entity that plays some role in that process. 
There is a ``linearization'' procedure
that converts a binary/relational logical representation 
({\ttfamily FOLE}, conceptual graphs)
to a unary/functional logical representation
(Sketches~\cite{johnson:rosebrugh:wood:era}, Ologs~\cite{spivak:kent:olog}).
In this example,
linearization would define \emph{functional} roles,
changing the ternary relation type (process)
to an entity type $\mathtt{Go}{\,\in\,}X$ and
converting its arity elements (participent roles)
to function types
$\mathtt{agnt}{\,\in\,}{\Omega}_{\mathtt{Person},{\langle{\mathbf{1},\mathtt{Go}}\rangle}}$,
$\mathtt{dest}{\,\in\,}{\Omega}_{\mathtt{City},{\langle{\mathbf{1},\mathtt{Go}}\rangle}}$ and
$\mathtt{inst}{\,\in\,}{\Omega}_{\mathtt{Bus},{\langle{\mathbf{1},\mathtt{Go}}\rangle}}$.}
In a conceptual graph representation,
the logic language $\mathcal{L} = {\langle{R,\sigma,X,\Omega}\rangle}$
corresponds to a CG module ${\langle{X,R,C}\rangle}$
with type hierarchy $X$, relation hierarchy $R$ and catalog of individuals $C{\,\subseteq\,}\Omega$.
A CG representation is
%
\begin{center}
{\ttfamily\scriptsize\begin{tabular}{@{\hspace{-50pt}}l}
[Go]-\\
\hspace{16pt}(agnt)->[Person: John]\\
\hspace{16pt}(dest)->[City: Boston]\\
\hspace{16pt}(inst)->[Bus].
\end{tabular}}
\end{center}
Formally (compare this linear form to~\ref{roles}),
we have the following elements:
four entity types $\mathtt{Go},\mathtt{Person},\mathtt{City},\mathtt{Bus}{\,\in\,}X$;
three relation types $\mathtt{agnt},\mathtt{dest},\mathtt{inst}{\,\in\,}R$
with signatures
$\sigma(\mathtt{agnt}) = {\langle{\mathbf{2},s_{\mathtt{agnt}}}\rangle}$,
$\sigma(\mathtt{dest}) = {\langle{\mathbf{2},s_{\mathtt{dest}}}\rangle}$,
$\sigma(\mathtt{inst}) = {\langle{\mathbf{2},s_{\mathtt{inst}}}\rangle}$
having valence 2, arity $\mathbf{2} = \{0,1\}$ and
signatures $s_{\mathtt{agnt}},s_{\mathtt{dest}},s_{\mathtt{inst}}:\mathbf{2} \rightarrow X$,
where
$s_{\mathtt{agnt}}(0)=s_{\mathtt{dest}}(0)=s_{\mathtt{inst}}(0)=\mathtt{Go}$,
$s_{\mathtt{agnt}}(1)=\mathtt{Person}$,
$s_{\mathtt{dest}}(1)=\mathtt{City}$, and
$s_{\mathtt{inst}}(1)=\mathtt{Bus}$;
and
two constants as above.
\mbox{}\newline
\item[Quantification:] 
The universal quantification `$\forall_{x\in{X}}P(x{:}X,y{:}Y,z{:}Z)$'
is traditionally viewed as formula flow 
along the type list inclusion $\{y,z\}\subseteq\{x,y,z\}$.
{\ttfamily FOLE} handles existential/universal quantification and substitution
in terms of formula flow (Table~\ref{lift:flow})
along type list morphisms 
in the relational aspect
or along term vectors
in the logical aspect.
Given a morphism of type lists 
${\langle{I',s'}\rangle} \xrightarrow{h} {\langle{I,s}\rangle}$,
for any table 
${\langle{K,t}\rangle} \in \mathrmbf{Tbl}_{\mathcal{E}}(I,s)$, 
you can get two tables
${\scriptstyle\sum}_{h}(K,t),{\scriptstyle\prod}_{h}(K,t) \in \mathrmbf{Tbl}_{\mathcal{E}}(I',s')$
as follows.  
Given any possible row (or better, tuple) $t' \in \mathrmbfit{tup}_{\mathcal{E}}(I',s')$,
you can ask either an existential or a universal question about it:
for example, 
``Does there \emph{exist} a key $k \in K$ in $T$ with image $t'$?'' 
($\mathrmbfit{tup}_{h}(t_{k}) = t'$)
or 
``Is it the case that \emph{all possible} tuples $t \in \mathrmbfit{tup}_{\mathcal{E}}(I,s)$
with image $t'$ are present in $T$?''
(\cite{kent:spivak:email})
\mbox{}\newline
\item[Relation/Database Joins:] 
The joins of $\mathcal{E}$-relations (or $\mathcal{E}$-tables) are represented in {\ttfamily FOLE} 
in terms of fibered products
 --- products modulo some reference.
If 
an $\mathcal{S}$-span of constraints
${\langle{I_{1},s_{1},\varphi_{1}}\rangle}\xleftarrow{h_{1}}{\langle{I,s,\varphi}\rangle}
\xrightarrow{h_{2}}{\langle{I_{2},s_{2},\varphi}\rangle}$
holds in a relational structure
$\mathcal{M} = 
{\langle{\mathcal{R},{\langle{\sigma,\tau}\rangle},\mathcal{E}}\rangle}$,
it is interpreted as an opspan of $\mathcal{E}$-relations 
(or $\mathcal{E}$-tables).
Then the join of $\mathcal{E}$-relations (or $\mathcal{E}$-tables)
is represent by the formula
${\iota_{1}}^{\ast}(\varphi_{1})\wedge_{{\langle{\widehat{I},\widehat{s}}\rangle}}{\iota_{2}}^{\ast}(\varphi_{2})$,
where 
${\langle{I_{1},s_{1}}\rangle}\xrightarrow{\iota_{1}}{\langle{\widehat{I},\widehat{s}}\rangle}
\xleftarrow{\iota_{2}}{\langle{I_{2},s_{2}}\rangle}$
is the fibered sum of type lists.
In general, 
the join of an arbitrary diagram 
of $\mathcal{E}$-relations (or $\mathcal{E}$-tables) 
is obtained by
substitution 
followed by 
conjunction.
\end{description}
%
%


\subsection{Logical Environment}\label{append:log-env}

Let
$\mathcal{S}_{2}={\langle{R_{2},\sigma_{2},X_{2}}\rangle} 
\xRightarrow{\langle{r,f}\rangle}
{\langle{R_{1},\sigma_{1},X_{1}}\rangle}=\mathcal{S}_{1}$
be a schema morphism,
with structure fiber passage
$\mathrmbf{Struc}(\mathcal{S}_{2}) \xleftarrow{\mathrmbfit{struc}_{{\langle{r,f}\rangle}}} \mathrmbf{Struc}(\mathcal{S}_{2})$
and
bridging structure morphism
\[\mbox{\footnotesize{
$\mathrmbfit{struc}_{{\langle{r,f}\rangle}}(\mathcal{M}_{1}) 
= {\langle{r^{-1}(\mathcal{R}_{1}),{\langle{\sigma_{2},\tau_{1}}\rangle},f^{-1}(\mathcal{E}_{1})}\rangle}
\stackrel{{\langle{r,1_{K},f,1_{Y}}\rangle}}{\rightleftarrows} 
{\langle{\mathcal{R}_{1},{\langle{\sigma_{1},\tau_{1}}\rangle},\mathcal{E}_{1}}\rangle} = \mathcal{M}_{1}$
}\normalsize}\]
with relation and entity infomorphisms
$r^{-1}(\mathcal{R}_{1})
\stackrel{{\langle{r,1_{K}}\rangle}}{\rightleftarrows}
\mathcal{R}_{1}$ 
and
$f^{-1}(\mathcal{E}_{1})
\stackrel{{\langle{f,1_{Y}}\rangle}}{\rightleftarrows} 
\mathcal{E}_{1}$.

\begin{proposition}
The (formula) interpretation of the inverse image structure
is the inverse image of the (formula) interpretation.
\end{proposition}
\begin{fact}
The formula classification of the inverse image relation classfication is
the inverse image classfication of the formula relation classification:
\[\mbox{\footnotesize{$
\widehat{r^{-1}(\mathcal{R}_{1})} 
= \widehat{{\langle{R_{2},K_{1},\models_{r}}\rangle}} 
= {\langle{\widehat{R}_{2},K_{1},\models_{\widehat{r}}}\rangle}
= \widehat{r}^{-1}(\widehat{\mathcal{R}}_{1}).
$}\normalsize}\]
\end{fact}
\begin{proof}
The proof is by induction on formulas $\varphi_{2} \in \widehat{R}_{2}$.
\end{proof}
\begin{fact}
The formula structure morphism of the bridging structure morphism is: 
\[\mbox{\footnotesize{$
{\langle{\widehat{r},1_{K},f,1_{Y}}\rangle} : 
{\langle{\widehat{r^{-1}(\mathcal{R}_{1})},{\langle{\sigma_{2},\tau_{1}}\rangle},f^{-1}(\mathcal{E}_{1})}\rangle} 
\rightleftarrows
{\langle{\widehat{\mathcal{R}}_{1},{\langle{\sigma_{1},\tau_{1}}\rangle},\mathcal{E}_{1}}\rangle}. 
$}\normalsize}\]
Its ($\mathrmbfit{inst}$-vertical) relation infomorphism
\newline
${\langle{\widehat{r},1_{K}}\rangle} : 
\widehat{r^{-1}(\mathcal{R}_{1})} 
= 
\widehat{{\langle{R_{2},K_{1},\models_{r}}\rangle}} 
= 
{\langle{\widehat{R}_{2},K_{1},\models_{\widehat{r}}}\rangle}
\rightleftarrows 
{\langle{\widehat{R}_{1},K_{1},\models_{\widehat{\mathcal{R}}_{1}}}\rangle} = \widehat{\mathcal{R}}_{1}$
\newline
is the bridging infomorphism of the formula relation classification, 
with the infomorphism condition
$k_{1}{\;\models_{\widehat{r^{-1}(\mathcal{R}_{1})}}\;}\varphi_{2}$
\underline{iff}
$k_{1}{\;\models_{\widehat{\mathcal{R}}_{1}}\;}\widehat{r}(\varphi_{2})$.
The extent monotonic function 
$\widehat{r} :
\mathrmbfit{ext}(\widehat{r^{-1}(\mathcal{R}_{1})})\rightarrow\mathrmbfit{ext}(\widehat{\mathcal{R}}_{1})$
is an isometry:
$\varphi{\;\leq_{\widehat{r^{-1}(\mathcal{R}_{1})}}\;}\psi$
iff
$\widehat{r}(\varphi){\;\leq_{\widehat{\mathcal{R}}_{1}}\;}\widehat{r}(\psi)$.
\end{fact}
\begin{proposition}
Satisfaction is invariant under change of notation;
that is,
for any schema morphism
$\mathcal{S}_{2}={\langle{R_{2},\sigma_{2},X_{2}}\rangle} 
\xRightarrow{\langle{r,f}\rangle}
{\langle{R_{1},\sigma_{1},X_{1}}\rangle}=\mathcal{S}_{1}$
the following satisfaction condition holds:
\[\mbox{\footnotesize{$
\mathrmbfit{struc}_{{\langle{r,f}\rangle}}(\mathcal{M}_{1})
{\;\models_{\mathcal{S}_{2}}\;}
(\varphi_{2}\xrightarrow{h}\varphi_{2}')
				\;\;\;\text{\underline{iff}}\;\;\;
\mathcal{M}_{1}
\;\models_{\mathcal{S}_{1}}\; 
(\widehat{r}(\varphi_{2})\xrightarrow{h}\widehat{r}(\varphi_{2}'))=
\mathrmbfit{fmla}_{{\langle{r,f}\rangle}}(\varphi_{2}{\;\vdash\;}\varphi_{2}')
$.}\normalsize}\]
\end{proposition}
\begin{proof}
But this holds, since
$\widehat{r^{-1}(\mathcal{R}_{1})} = \widehat{r}^{-1}(\widehat{\mathcal{R}}_{1})$.
%
In more detail,
\mbox{}\newline
$\mathrmbfit{struc}_{{\langle{r,f}\rangle}}(\mathcal{M}_{1})
{\;\models_{\mathcal{S}_{2}}\;}
(\varphi_{2}\xrightarrow{h}\varphi_{2}')$
\underline{iff}
${\scriptstyle\sum}_{h}(\varphi_{2}'){\;\leq_{\widehat{r^{-1}(\mathcal{R}_{1})}}\;}\varphi_{2}$
\newline
\underline{iff}
$\widehat{r}({\scriptstyle\sum}_{h}(\varphi_{2}')){\;\leq_{\widehat{\mathcal{R}}_{1}}\;}\widehat{r}(\varphi_{2})$
\underline{iff}
${\scriptstyle\sum}_{h}(\widehat{r}(\varphi_{2}')){\;\leq_{\widehat{\mathcal{R}}_{1}}\;}\widehat{r}(\varphi_{2})$
\newline
\underline{iff}
$\mathcal{M}_{1}
\;\models_{\mathcal{S}_{1}}\; 
(\widehat{r}(\varphi_{2})\xrightarrow{h}\widehat{r}(\varphi_{2}'))=
\mathrmbfit{fmla}_{{\langle{r,f}\rangle}}(\varphi_{2}{\;\vdash\;}\varphi_{2}')$.
\end{proof}
\begin{proposition}
The institution
${\langle{\mathrmbf{Sch},\mathrmbfit{fmla},\mathrmbfit{struc}}\rangle}$ 
is a logical environment,
since it satisfies the bimodular principle
``satisfaction respects structure morphisms'':
given any schema $\mathcal{S} = {\langle{R,\sigma,X}\rangle}$,
if
${\langle{1_{R},k,1_{X},g}\rangle} : \mathcal{M}_{2} \rightleftarrows \mathcal{M}_{1}$
is a $\mathrmbfit{sch}$-vertical structure morphism over $\mathcal{S}$, 
then
we have the intent order
$\mathcal{M}_{2} \geq_{\mathcal{S}} \mathcal{M}_{1}$;
that is,
$\mathcal{M}_{2} \models_{\mathcal{S}} (\varphi{\;\vdash\;}\psi)$ 
implies 
$\mathcal{M}_{1} \models_{\mathcal{S}} (\varphi{\;\vdash\;}\psi)$ 
for any $\mathcal{S}$-sequent $(\varphi{\;\vdash\;}\psi)$.
\footnote{For any classification
$\mathcal{A} = {\langle{X,Y,\models_{\mathcal{A}}}\rangle}$, 
the intent order 
$\mathrmbfit{int}(\mathcal{A}) = {\langle{Y,\leq_{\mathcal{A}}}\rangle}$
is defined as follows:
for two instances $y,y'{\,\in\,}Y$,
$y{\;\leq_{\mathcal{A}}\;}y'$
when
$\mathrmbfit{int}_{\mathcal{A}}(y){\;\supseteq\;}\mathrmbfit{int}_{\mathcal{A}}(y')$;
that is,
when
$y'{\;\models_{\mathcal{A}}\;}x$ implies $y{\;\models_{\mathcal{A}}\;}x$
for each $x{\,\in\,}X$.}
\end{proposition}
\begin{proof}
The $\mathrmbfit{typ}$-vertical formula morphism
${\langle{1_{\widehat{R}},k,1_{X},g}\rangle} : \widehat{\mathcal{M}}_{2} \rightleftarrows \widehat{\mathcal{M}}_{1}$
over $\widehat{\mathcal{S}}$
\newline
has
the $\mathrmbfit{typ}$-vertical relation infomorphism
${\langle{1_{\widehat{R}},k}\rangle} : 
\widehat{\mathcal{R}}_{2}
\rightleftarrows 
\widehat{\mathcal{R}}_{1}$
over $\widehat{R}$.
\newline
$\mathcal{M}_{2} \models_{\mathcal{S}} (\varphi{\;\vdash\;}\psi)$ 
\underline{iff} 
$\varphi{\;\leq_{\widehat{\mathcal{R}}_{2}}\;}\psi$
\underline{implies} 
$\varphi{\;\leq_{\widehat{\mathcal{R}}_{1}}\;}\psi$
\underline{iff} 
$\mathcal{M}_{1} \models_{\mathcal{S}} (\varphi{\;\vdash\;}\psi)$ 
\newline
for any $\mathcal{S}$-sequent $(\varphi{\;\vdash\;}\psi)$.
\end{proof}
%

\subsection{Transformation to Databases}\label{append:cls2db}

\subsubsection{Relational Interpretation.}\label{rel:interp}
Let
$\mathcal{M} = {\langle{\mathcal{R},{\langle{\sigma,\tau}\rangle},\mathcal{E}}\rangle}$ 
be a (model-theoretic) relational structure.
The relation classification $\mathcal{R}$ is equivalent to the extent function
$\mathrmbfit{ext}_{\mathcal{R}} : R \rightarrow {\wp}K$,
which maps a relational symbol $r \in R$ to its $\mathcal{R}$-extent
$\mathrmbfit{ext}_{\mathcal{R}}(r) \subseteq K$. 
The list classification $\mathrmbf{List}(\mathcal{E})$ is equivalent to the extent function
$\mathrmbfit{ext}_{\mathrmbf{List}(\mathcal{E})} : 
\mathrmbf{List}(X) \rightarrow {\wp}\mathrmbf{List}(Y)$,
a restriction of the tuple passage
$\mathrmbfit{tup}_{\mathcal{E}} : 
\mathrmbf{List}(X)^{\mathrm{op}} \rightarrow \mathrmbf{Set}$,
which maps a type list ${\langle{I,s}\rangle} \in \mathrmbf{List}(X)$ to its $\mathrmbf{List}(\mathcal{E})$-extent
$\mathrmbfit{tup}_{\mathcal{E}}(I,s) \subseteq \mathrmbf{List}(Y)$.
The list designation satisfies the condition
$k \models_{\mathcal{R}} r$ implies $\tau(k) \models_{\mathrmbf{List}(\mathcal{E})} \sigma(r)$
for all $k \in K$ and $r \in R$;
so that
$k \in \mathrmbfit{ext}_{\mathcal{R}}(r)$
implies
$\tau(k) \in \mathrmbfit{ext}_{\mathrmbf{List}(\mathcal{E})}(\sigma(r)) = \mathrmbfit{tup}_{\mathcal{E}}(\sigma(r))$.
Hence,
${\wp}\tau(\mathrmbfit{ext}_{\mathcal{R}}(r)) \subseteq \mathrmbfit{tup}_{\mathcal{E}}(\sigma(r))$
for all $r \in R$.
Thus,
we have the function order
${\mathrmbfit{ext}_{\mathcal{R}} \cdot {\wp}\tau} \subseteq 
{\sigma \cdot \mathrmbfit{ext}_{\mathrmbf{List}(\mathcal{E})}}$. 
%

The relational interpretation function
$\mathrmbfit{R}_{\mathcal{M}} : R \rightarrow |\mathrmbf{Rel}(\mathcal{E})|$
\comment{
\footnote{For relational structure
$\mathcal{M} = {\langle{\mathcal{R},{\langle{\sigma,\tau}\rangle},\mathcal{E}}\rangle}$,
the fibered mathematical context $\mathrmbf{Rel}(\mathcal{E})^{\mathrm{op}}\xrightarrow{\mathrmbfit{list}}\mathrmbf{List}(X)$ 
is determined by 
the indexed preorder
$\mathrmbf{List}(X)^{\mathrm{op}}\xrightarrow{\mathrmbfit{rel}}\mathrmbf{Pre}$,
which maps a type list ${\langle{I,s}\rangle}$ to the fiber relational order
$\mathrmbf{Rel}_{\mathcal{E}}(I,s)
={\wp}\mathrmbfit{tup}_{\mathcal{E}}(I,s)$
and maps a type list morphism
${\langle{I,s}\rangle} \xrightharpoondown{h} {\langle{I',s'}\rangle}$
to the fiber monotonic function
${\exists}_{h} = {\exists}_{\mathrmbfit{tup}_{\mathcal{E}}(h)} :
\mathrmbf{Rel}_{\mathcal{E}}(I,s)\leftarrow\mathrmbf{Rel}_{\mathcal{E}}(I',s')$.
Similarly, for
$\mathrmbf{Tbl}(\mathcal{E})^{\mathrm{op}}\xrightarrow{\mathrmbfit{pr}}\mathrmbf{List}(X)$.}
}
maps a relational symbol $r\in{R}$
with type list $\sigma(r)={\langle{I,s}\rangle}$
to the set of tuples
$\mathrmbfit{R}_{\mathcal{M}}(r)={\wp}\tau(\mathrmbfit{ext}_{\mathcal{R}}(r))
\in {\wp}\mathrmbfit{tup}_{\mathcal{E}}(I,s)=\mathrmbf{Rel}_{\mathcal{E}}(I,s)$.
The tabular interpretation function
$\mathrmbfit{T}_{\mathcal{M}} : R \rightarrow 
|\mathrmbf{Tbl}(\mathcal{E})| = |{(\mathrmbf{Set}{\downarrow}\mathrmbfit{tup}_{\mathcal{E}})}|$
maps a relational symbol $r\in{R}$
with type list $\sigma(r)={\langle{I,s}\rangle}$
to the pair
$\mathrmbfit{T}_{\mathcal{M}}(r) = {\langle{K(r),t_{r}}\rangle}$
consisting of 
the key set $K(r) = \mathrmbfit{ext}_{\mathcal{R}}(r) \subseteq K$ and 
the tuple function $K(r)\xrightarrow{t_{r}}\mathrmbfit{tup}_{\mathcal{E}}(I,s)$,
a restriction of the 
tuple function $\tau : K \rightarrow \mathrmbf{List}(Y)$,
which maps a key $k \in K_{r}$
to the tuple $t_{r}(k) = \tau(k) \in \mathrmbfit{tup}_{\mathcal{E}}(I,s)$.
Applying the image passage
$\mathrmbfit{im}_{\mathcal{E}}(I,s) : \mathrmbf{Tbl}_{\mathcal{E}}(I,s) \rightarrow \mathrmbf{Rel}_{\mathcal{E}}(I,s)$,
the image of the table interpretation is the relation interpretation
$\mathrmbfit{im}_{\mathcal{E}}(I,s)(\mathrmbfit{T}_{\mathcal{M}}(r)) = \mathrmbfit{R}_{\mathcal{M}}(r)$
for any relation symbol $r \in R$.
Using the combined image passage
$\mathrmbfit{im}_{\mathcal{E}} : \mathrmbf{Tbl}(\mathcal{E}) \rightarrow \mathrmbf{Rel}(\mathcal{E})$,
we get the composition
$\mathrmbfit{R}_{\mathcal{M}} =
R \xrightarrow{\mathrmbfit{T}_{\mathcal{M}}} 
|\mathrmbf{Tbl}(\mathcal{E})| \xrightarrow{|\mathrmbfit{im}_{\mathcal{E}}|} |\mathrmbf{Rel}(\mathcal{E})|$.
Note that
$t_{r} : K_{r}\rightarrow\mathrmbfit{R}_{\mathcal{M}}(r) \rightarrow \mathrmbfit{tup}_{\mathcal{E}}(I,s)$,
is a surjection-injection factorization of the tuple function. 
\footnote{
Two tables are informationally equivalent
when they contain the same information; that is,
when their image relations are equivalent in 
$\mathrmbf{Rel}_{\mathcal{E}}(I,s)={\wp}\mathrmbfit{tup}_{\mathcal{E}}(I,s)$.
In particular,
the table $\mathrmbfit{T}_{\mathcal{M}}(r)$ and relation $\mathrmbfit{R}_{\mathcal{M}}(r)$
of a relational symbol are informationally equivalent.
}
%

\subsubsection{Relational Logics.}\label{rel:log}
%
A relational logic 
$\mathcal{L}={\langle{\mathcal{S},\mathcal{M},T}\rangle}$
consists of 
a relational structure $\mathcal{M} = {\langle{\mathcal{R},{\langle{\sigma,\tau}\rangle},\mathcal{E}}\rangle}$ and
a relational specification $\mathcal{T}={\langle{\mathcal{S},T}\rangle}$
that share a common relational schema 
$\mathrmbfit{sch}(\mathcal{M}) = \mathcal{S}$.
The logic is sound when
the structure $\mathcal{M}$ satisfies every constraint in the specification $T$.
A sound relational logic enriches a relational structure with a specification.
For any sound logic 
$\mathcal{L}={\langle{\mathcal{S},\mathcal{M},T}\rangle}$,
there is an interpretation functor
$\widehat{\mathrmbf{R}}^{\mathrm{op}}
\xrightarrow{\mathrmbfit{T}_{\mathcal{L}}}
\mathrmbf{Tbl}(\mathcal{E}) = {(\mathrmbf{Set}{\downarrow}\mathrmbfit{tup}_{\mathcal{E}})}$,
where
$\widehat{\mathrmbf{R}}\subseteq\mathrmbf{Fmla}(\mathcal{S})$
is the consequence of $T$.
Sound logics are important in the transformation of structures to databases (below).
%
%
A relational logic morphism
$\mathcal{L}_{2}={\langle{\mathcal{S}_{2},\mathcal{M}_{2},T_{2}}\rangle}
\xrightarrow{\langle{{\langle{r,k}\rangle},{\langle{f,g}\rangle}}\rangle}
{\langle{\mathcal{S}_{2},\mathcal{M}_{2},T_{2}}\rangle}=\mathcal{L}_{2}$
consists of 
a relational structure morphism
$\mathcal{M}_{2}
\xrightarrow{\langle{{\langle{r,k}\rangle},{\langle{f,g}\rangle}}\rangle}
\mathcal{M}_{1}$ and
a relational specification morphism
$\mathcal{T}_{2} = {\langle{\mathcal{S}_{2},T_{2}}\rangle}
\xrightarrow{\langle{r,f}\rangle}
{\langle{\mathcal{S}_{1},T_{1}}\rangle} = \mathcal{T}_{1}$
that share a common relational schema morphism  
$\mathrmbfit{sch}({\langle{r,k}\rangle},{\langle{f,g}\rangle})=
\mathcal{S}_{2}
\stackrel{{\langle{r,f}\rangle}}{\Longrightarrow}
\mathcal{S}_{1}$.
\comment{
Let 
\[\mbox{\footnotesize{$
\mathrmbf{Struc}\xleftarrow{\mathrmbfit{struc}}\hspace{-20pt}
\underset{\mathrmbf{Struc}{\times}_{\mathrmbf{Lang}}\mathrmbf{Spec}}{\mathrmbf{Log}}
\hspace{-20pt}\xrightarrow{\mathrmbfit{spec}}\mathrmbf{Spec}
$}\normalsize}\]
denote the mathematical context of relational logics
and let $\mathrmbf{Snd}\subseteq\mathrmbf{Log}$ denote the subcontext of sound logics. 
}

Any sound relational logic 
$\mathcal{L}={\langle{\mathcal{S},\mathcal{M},T}\rangle}$
with
structure $\mathcal{M} = {\langle{\mathcal{R},{\langle{\sigma,\tau}\rangle},\mathcal{E}}\rangle}$ and
specification $\mathcal{T}={\langle{\mathcal{S},T}\rangle}$
has an associated logical/relational database
$\mathrmbfit{db}(\mathcal{L}) = {\langle{\mathcal{S},\mathcal{E},\mathrmbfit{K},\tau}\rangle}$
with 
category of formulas $\widehat{\mathrmbf{R}}\subseteq\mathrmbf{Fmla}(\mathcal{S})$ 
(the consequence of $T$), 
signature passage $\mathrmbfit{S} : \widehat{\mathrmbf{R}} \rightarrow \mathrmbf{List}(X)$,
entity classification $\mathcal{E}$, 
key passage $\mathrmbfit{K} : \widehat{\mathrmbf{R}}^{\mathrm{op}} \rightarrow \mathrmbf{Set}$,
tuple bridge $\tau : \mathrmbfit{K} \Rightarrow \mathrmbfit{S}^{\mathrm{op}} \circ \mathrmbfit{tup}_{\mathcal{E}}$, and
table interpretation passage 
$\widehat{\mathrmbf{R}}^{\mathrm{op}}
\xrightarrow{\mathrmbfit{T}} \mathrmbf{Tbl}(\mathcal{E}) = {(\mathrmbf{Set}{\downarrow}\mathrmbfit{tup}_{\mathcal{E}})}$,
where $\tau = \mathrmbfit{T}\tau_{\mathcal{E}}$. 
Any sound relational logic morphism
$\mathcal{L}_{2}={\langle{\mathcal{S}_{2},\mathcal{M}_{2},T_{2}}\rangle}
\xrightarrow{\langle{{\langle{r,k}\rangle},{\langle{f,g}\rangle}}\rangle}
{\langle{\mathcal{S}_{2},\mathcal{M}_{2},T_{2}}\rangle}=\mathcal{L}_{2}$
with 
structure morphism
$\mathcal{M}_{2}
\xrightarrow{\langle{{\langle{r,k}\rangle},{\langle{f,g}\rangle}}\rangle}
\mathcal{M}_{1}$ and
specification morphism
$\mathcal{T}_{2} = {\langle{\mathcal{S}_{2},T_{2}}\rangle}
\xrightarrow{\langle{r,f}\rangle}
{\langle{\mathcal{S}_{1},T_{1}}\rangle} = \mathcal{T}_{1}$
has an associated (strict) logical/relational database morphism
$\mathrmbfit{db}({\langle{r,k}\rangle},{\langle{f,g}\rangle}) = 
{\langle{\mathrmbfit{F},f,g,\kappa}\rangle} :
\mathrmbfit{db}(\mathcal{L}_{2}) = {\langle{\mathcal{S}_{2},\mathcal{E}_{2},\mathrmbfit{K}_{2},\tau_{2}}\rangle} \rightarrow
{\langle{\mathcal{S}_{1},\mathcal{E}_{1},\mathrmbfit{K}_{1},\tau_{1}}\rangle} = \mathrmbfit{db}(\mathcal{L}_{1})$
with
(strict) database schema morphism
${\langle{\mathrmbfit{F},f}\rangle} : \mathcal{S}_{2} \rightarrow \mathcal{S}_{1}$,
entity infomorphism 
${\langle{f,g}\rangle} : 
\mathcal{E}_{2}
\rightleftarrows 
\mathcal{E}_{1}$, and 
key natural transformation
$\kappa : \mathrmbfit{F}^{\mathrm{op}} \circ \mathrmbfit{K}_{1} \Rightarrow \mathrmbfit{K}_{2}$,
which satisfy the condition
$\kappa{\;\bullet\;}\tau_{2}
= \mathrmbfit{F}^{\mathrm{op}}\tau_{1}{\;\bullet\;}\mathrmbfit{S}_{2}^{\mathrm{op}}\tau_{{\langle{f,g}\rangle}}$.
The passage
$\widehat{\mathrmbf{R}}_{2}\xrightarrow{\mathrmbfit{F}}\widehat{\mathrmbf{R}}_{1}$
from formula subcontext
$\widehat{\mathrmbf{R}}_{2}\subseteq\mathrmbf{Fmla}(\mathcal{S}_{2})$
to formula subcontext
$\widehat{\mathrmbf{R}}_{1}\subseteq\mathrmbf{Fmla}(\mathcal{S}_{1})$
is a restriction of 
the fibered formula passage
$\mathrmbf{Fmla}(\mathcal{S}_{2}) \xrightarrow{\mathrmbfit{fmla}_{{\langle{r,f}\rangle}}} \mathrmbf{Fmla}(\mathcal{S}_{1})$.
(Kent~\cite{kent:db:sem} has more details on relational database semantics.) 
\begin{figure}
\begin{center}
\begin{tabular}{c}
{\footnotesize{$\mathcal{L}_{2}={\langle{\mathcal{S}_{2},\mathcal{M}_{2},T_{2}}\rangle}
\xrightarrow{\langle{{\langle{r,k}\rangle},{\langle{f,g}\rangle}}\rangle}
{\langle{\mathcal{S}_{2},\mathcal{M}_{2},T_{2}}\rangle}=\mathcal{L}_{2}$}}
\\
\\
\setlength{\unitlength}{0.55pt}
\begin{picture}(120,240)(0,0)
\put(52,236){{\Large{$\Downarrow$}}\hspace{5pt}{\normalsize{$\mathrmbfit{db}$}}}
\put(5,200){\makebox(0,0){\footnotesize{$\widehat{\mathrmbf{R}}_{2}^{\mathrm{op}}$}}}
\put(125,200){\makebox(0,0){\footnotesize{$\widehat{\mathrmbf{R}}_{1}^{\mathrm{op}}$}}}
\put(60,210){\makebox(0,0){\scriptsize{$\mathrmbfit{F}^{\mathrm{op}}$}}}
\put(20,200){\vector(1,0){80}}
\put(0,185){\vector(0,-1){90}}
\put(-16,188){\vector(-4,-3){48}}
\put(-64,128){\vector(4,-3){48}}
\put(-44,174){\makebox(0,0)[r]{\scriptsize{$\mathrmbfit{T}_{2}$}}}
\put(6,140){\makebox(0,0)[l]{\scriptsize{$\mathrmbfit{S}_{2}^{\mathrm{op}}$}}}
\put(-132,120){\makebox(0,0)[r]{\scriptsize{$\mathrmbfit{K}_{2}$}}}
\put(-44,108){\makebox(0,0)[r]{\scriptsize{$\mathrmbfit{sign}^{\mathrm{op}}_{\mathcal{E}_{2}}$}}}
\put(-67,60){\makebox(0,0)[r]{\shortstack{\scriptsize{$\mathrmbfit{key}_{\mathcal{E}_{2}}$}}}}
\put(120,185){\vector(0,-1){90}}
\put(136,188){\vector(4,-3){48}}
\put(184,128){\vector(-4,-3){48}}
\qbezier(-20,200)(-130,180)(-130,120)\qbezier(-130,120)(-130,30)(35,3)\put(35,3){\vector(4,-1){0}}
\qbezier(140,200)(250,180)(250,120)\qbezier(250,120)(250,30)(85,3)\put(85,3){\vector(-4,-1){0}}
\qbezier(-84,125)(-80,20)(40,10)\put(40,10){\vector(1,0){0}}
\qbezier(204,125)(200,20)(80,10)\put(80,10){\vector(-1,0){0}}
\put(60,148.5){\makebox(0,0){\large{$=$}}}
\put(35,115){\makebox(0,0){\large{$\overset{\rule[-2pt]{0pt}{5pt}\kappa}{\Longleftarrow}$}}}
\qbezier[150](-128,110)(60,110)(248,110)
\put(164,174){\makebox(0,0)[l]{\scriptsize{$\mathrmbfit{T}_{1}$}}}
\put(114,140){\makebox(0,0)[r]{\scriptsize{$\mathrmbfit{S}_{1}^{\mathrm{op}}$}}}
\put(252,120){\makebox(0,0)[l]{\scriptsize{$\mathrmbfit{K}_{1}$}}}
\put(164,108){\makebox(0,0)[l]{\scriptsize{$\mathrmbfit{sign}^{\mathrm{op}}_{\mathcal{E}_{1}}$}}}
\put(192,60){\makebox(0,0)[l]{\shortstack{\scriptsize{$\mathrmbfit{key}_{\mathcal{E}_{1}}$}}}}
\put(-37,60){\makebox(0,0){\shortstack{\scriptsize{$\tau_{\mathcal{E}}$}\\\large{$\Rightarrow$}}}}
\put(157,60){\makebox(0,0){\shortstack{\scriptsize{$\;\;\tau_{\mathcal{E}_{1}}$}\\\large{$\Leftarrow$}}}}
\put(-80,140){\makebox(0,0){\footnotesize{$\mathrmbf{Tbl}(\mathcal{E}_{2})$}}}
\put(200,140){\makebox(0,0){\footnotesize{$\mathrmbf{Tbl}(\mathcal{E}_{1})$}}}
\put(0,80){\makebox(0,0){\footnotesize{$\mathrmbf{List}(X_{2})^{\mathrm{op}}$}}}
\put(124,80){\makebox(0,0){\footnotesize{$\mathrmbf{List}(X_{1})^{\mathrm{op}}$}}}
\put(60,5){\makebox(0,0){\footnotesize{$\mathrmbf{Set}$}}}
\put(67,90){\makebox(0,0){\scriptsize{$({\scriptstyle\sum}_{f})^{\mathrm{op}}$}}}
\put(24,38){\makebox(0,0)[r]{\scriptsize{$\mathrmbfit{tup}_{\mathcal{E}_{2}}$}}}
\put(97,38){\makebox(0,0)[l]{\scriptsize{$\mathrmbfit{tup}_{\mathcal{E}_{1}}$}}}
	\put(60,55){\makebox(0,0){\shortstack{\scriptsize{$\;\;\;{\tau}_{{\langle{f,g}\rangle}}$}\\\large{$\Leftarrow$}}}}
\qbezier[25](25,48)(60,48)(95,48)
\put(35,80){\vector(1,0){50}}
\put(9,68){\vector(3,-4){40}}
\put(111,68){\vector(-3,-4){40}}
%
%
\end{picture}
\\
\\
{\footnotesize $\mathrmbfit{db}(\mathcal{L}_{2}) = {\langle{\mathcal{S}_{2},\mathcal{E}_{2},\mathrmbfit{K}_{2},\tau_{2}}\rangle} 
\xrightarrow{\langle{\mathrmbfit{F},f,g,\kappa}\rangle}
{\langle{\mathcal{S}_{1},\mathcal{E}_{1},\mathrmbfit{K}_{1},\tau_{1}}\rangle} = \mathrmbfit{db}(\mathcal{L}_{1})$}
\\ \\
{\footnotesize 
$\kappa \bullet \tau_{2}
= \mathrmbfit{F}^{\mathrm{op}}\tau_{1} \bullet 
\mathrmbfit{S}_{2}^{\mathrm{op}}\tau_{{\langle{f,g}\rangle}}$}
\end{tabular}
\end{center}
\caption{From Sound Logics to Logical/Relational Databases}
\label{fig:snd2db}
\end{figure}
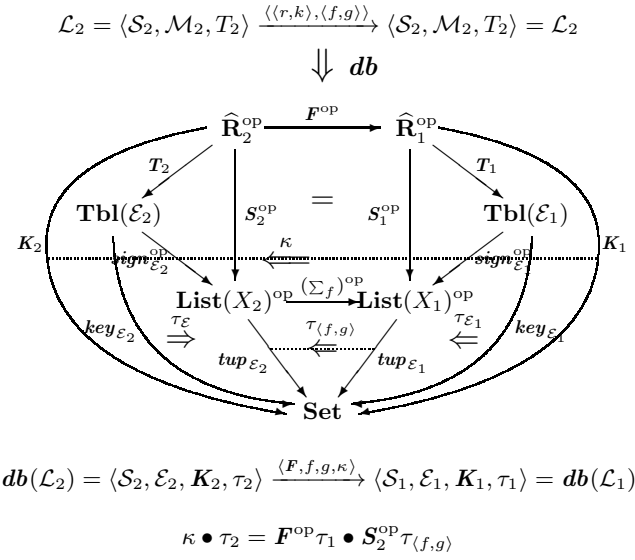
%

\end{document}